\documentclass[journal]{IEEEtran}
\usepackage{amssymb}
\usepackage[cmex10]{amsmath}
\usepackage{stfloats}
\usepackage{graphicx}
\usepackage{multirow}
\usepackage{subfigure}
\usepackage{tabularx}
\usepackage{booktabs}
\usepackage{amsmath}
\usepackage{epsfig,epsf,color,balance,cite}
\usepackage{indentfirst}
\usepackage{mathrsfs}
\usepackage[noend]{algpseudocode}
\usepackage{algorithmicx,algorithm}

\newtheorem{theorem}{Theorem}
\newtheorem{lemma}{Lemma}
\newtheorem{remark}{Remark}

\begin{document}
% paper title
\title{On the Secrecy Performance of Random VLC Networks with Imperfect CSI and Protected Zone}

\author{Jin-Yuan Wang, Yu Qiu, Sheng-Hong Lin, Jun-Bo Wang, Min Lin, and Cheng Liu
\thanks{This work was supported in part by National Natural Science Foundation of China under Grants 61701254 and 61571115, the Natural
Science Foundation of Jiangsu Province under Grant BK20170901, the Key International Cooperation Research Project under Grant 61720106003,
and the open research fund of Key Lab of Broadband Wireless Communication and Sensor Network Technology (Nanjing University of Posts and Telecommunications), the Ministry of Education under Grants JZNY201706 and JZNY201701. \emph{(Corresponding author: Jin-Yuan Wang.)}}
\thanks{Jin-Yuan Wang, Sheng-Hong Lin, and Min Lin are with Key Lab of Broadband Wireless Communication and Sensor Network Technology, Nanjing University of Posts and Telecommunications, Nanjing 210003, China. (E-mail: jywang@njupt.edu.cn, linshenghong5061@163.com, linmin@njupt.edu.cn)}
\thanks{Yu Qiu and Jun-Bo Wang are with National Mobile Communications Research Laboratory, Southeast University, Nanjing 210096, China. (E-mail: 220170896@seu.edu.cn, jbwang@seu.edu.cn)}
\thanks{Cheng Liu is with Vertical Industry Development Department, Baidu China Co., Ltd., Shanghai 201210, China. (E-mail: 220160880@seu.edu.cn)}
}

\maketitle
\begin{abstract}
This paper investigates the physical-layer security for a random indoor visible light communication (VLC) network with imperfect channel state information (CSI) and a protected zone. The VLC network consists of three nodes, i.e., a transmitter (Alice), a legitimate receiver (Bob), and an eavesdropper (Eve). Alice is fixed in the center of the ceiling, and the emitted signal at Alice satisfies the non-negativity and the dimmable average optical intensity constraint. Bob and Eve are randomly deployed on the receiver plane. By employing the protected zone and considering the imperfect CSI, the stochastic characteristics of the channel gains for both the main and the eavesdropping channels is first analyzed. After that, the closed-form expressions of the average secrecy capacity and the lower bound of secrecy outage probability are derived, respectively. Finally, Monte-Carlo simulations are provided to verify the accuracy of the derived theoretical expressions. Moreover, the impacts of the nominal optical intensity, the dimming target, the protected zone and the imperfect CSI on secrecy performance are discussed, respectively.
\end{abstract}


\IEEEpeerreviewmaketitle

\section{Introduction}
\label{section1}
Recent advances in light-emitting diodes (LEDs) are deriving a resurgence into the use of visible-light for wireless communications. This novel technology is called visible light communications (VLC) \cite{BIB01}. While the radio frequency wireless communications (RFWC) serve outdoor users or fast moving vehicular users, VLC can serve indoor users in future fifth generation (5G) wireless communications. Therefore, in future 5G technologies, VLC is considered to be a compelling technology for supplementing RFWC \cite{BIB02}.

To improve system performance, a lot of work has been presented in the last decade to investigate VLC in terms of channel modelling \cite{BIB03}, channel capacity analysis \cite{BIB04}, modulation \cite{BIB05}, coding \cite{BIB06}, indoor positioning \cite{BIB07}, underwater communication \cite{BIB08}, and hardware design \cite{BIB09}. Although extensive work on VLC has been performed, most of them focus on the point-to-point (P2P) communications. Nowadays, the research focus is being changed from P2P communications to networking aspects. Under VLC network scenarios, the data privacy and confidentiality are becoming more and more important for users. However, in many practical VLC scenarios, multiple LEDs are often employed for better illumination and thus the eavesdroppers are able to perform interception as long as they are in the area illuminated by one or multiple LEDs. Moreover, even when eavesdroppers are not allowed to access the specified area, they may still be capable of intercepting the information through the structure of the physical environment, such as keyholes and windows \cite{BIB10}. Therefore, the security is still a crucial issue for VLC, even if it is more advantageous than RFWC in terms of inherent security.

Recently, the physical-layer security (PLS) has emerged as a promising solution to protect information delivery from eavesdropping. In RFWC, the PLS has been extensively investigated. However, the derived results in RFWC cannot be directly applied to VLC for the following reasons \cite{BIB11}: (a) The transmitted optical intensity signal in VLC must be non-negative, while the signals in RFWC are with bipolar complex values; (b) The illumination and communication are simultaneously implemented in VLC, while only communication is considered in RFWC; and (c) the dimming control is an important consideration in VLC for power savings and energy efficiency, which is not considered in RFWC. Compared to the numerous studies in RFWC, few work has been presented for PLS in VLC. As a fundamental performance indicator, the secrecy capacity of VLC has been investigated \cite{BIB12,BIB13,BIB14,BIB15,BIB16}. In \cite{BIB12}, the secrecy capacity for a direct current based multiple input single output (MISO) VLC is analyzed, and a uniform input distribution is employed to derive the secrecy capacity bound. By using a truncated generalized normal (TGN) input distribution, the secrecy capacity and secure beamforming for MISO VLC are investigated in \cite{BIB13} and \cite{BIB14}, respectively. Considering the signal constraints of VLC, the uniform and TGN input distributions are generally not optimal. By using the variational method, an improved input distribution can be obtained \cite{BIB04}. By employing the variational method, three lower bounds of secrecy capacity for VLC are obtained in \cite{BIB15}, where the upper bounds are obtained in \cite{BIB16}. Note that refs. \cite{BIB12,BIB13,BIB14,BIB15,BIB16} focus on the fixed transceivers. To achieve a better understanding of the inherent security capabilities of VLC, more practical conditions (such as the distribution of randomly deployed receivers) should be considered. By employing the uniformly distributed receivers, the secrecy outage probability (SOP) for the hybrid VLC/RFWC system is studied in \cite{BIB17}. With spatially random terminals in VLC, the closed-form expressions for the SOP and the average secrecy capacity (ASC) are derived in \cite{BIB18}. In \cite{BIB19} and \cite{BIB20}, the SOP and beamforming are analyzed respectively for VLC with randomly located eavesdroppers. However, the dimming requirement for indoor VLC is not considered in \cite{BIB17,BIB18,BIB19,BIB20}. Moreover, the aforementioned works are carried out by assuming that the channel state information (CSI) is perfect. In practice, it is difficult to obtain perfect CSI because of channel estimation and quantization errors.

In this paper, the PLS for indoor VLC is further investigated. A more realistic three-node secure VLC network is established, which includes a transmitter, a legitimate receiver, and an eavesdropper. At the transmitter, the transmit visible-light signal satisfies the non-negativity and the dimmable average optical intensity constraint. To improve the secure performance, an eavesdropper-free protected zone is employed. Both the legitimate receiver and the eavesdropper are uniformly distributed. Moreover, the imperfect CSI at the receiver (CSIR) is considered for both the main and the eavesdropping channels. The main contribution of this paper is the derivation of closed-form expressions of the ASC and the lower bound of SOP for the VLC network with the imperfect CSIR and the protected zone. Based on the derived theoretical expressions, some interesting insights are provided. Finally, Monte-Carlo simulations are provided to verify the accuracy of the derived theoretical expressions.

The rest of this paper is organized as follows.
Section \ref{section2} presents the system model.
In Section \ref{section3}, the random channel characteristics is shown.
Then, Section \ref{section4} and Section \ref{section5} derive the closed-from expressions of the ASC and the lower bound of SOP, respectively. Numerical results are provided in Section \ref{section6}. Finally, Section \ref{section7} concludes the paper.

\section{System Model}
\label{section2}
As illustrated in Fig. \ref{fig1}, we consider an indoor VLC network, which includes a transmitter (i.e., Alice),
a legitimate receiver (i.e., Bob), and an eavesdropping receiver (i.e., Eve).
For Alice, one LED is employed as the lighting source,
which is installed in the center of the ceiling to ensure the illumination and transfer information to Bob.
When Alice transmits information to Bob, Eve as an eavesdropper can also receive the information.
At both Bob and Eve, the photodiode (PD) is employed to receive the information and perform the optical-to-electrical convention.

\begin{figure}
\centering
\includegraphics[width=7cm]{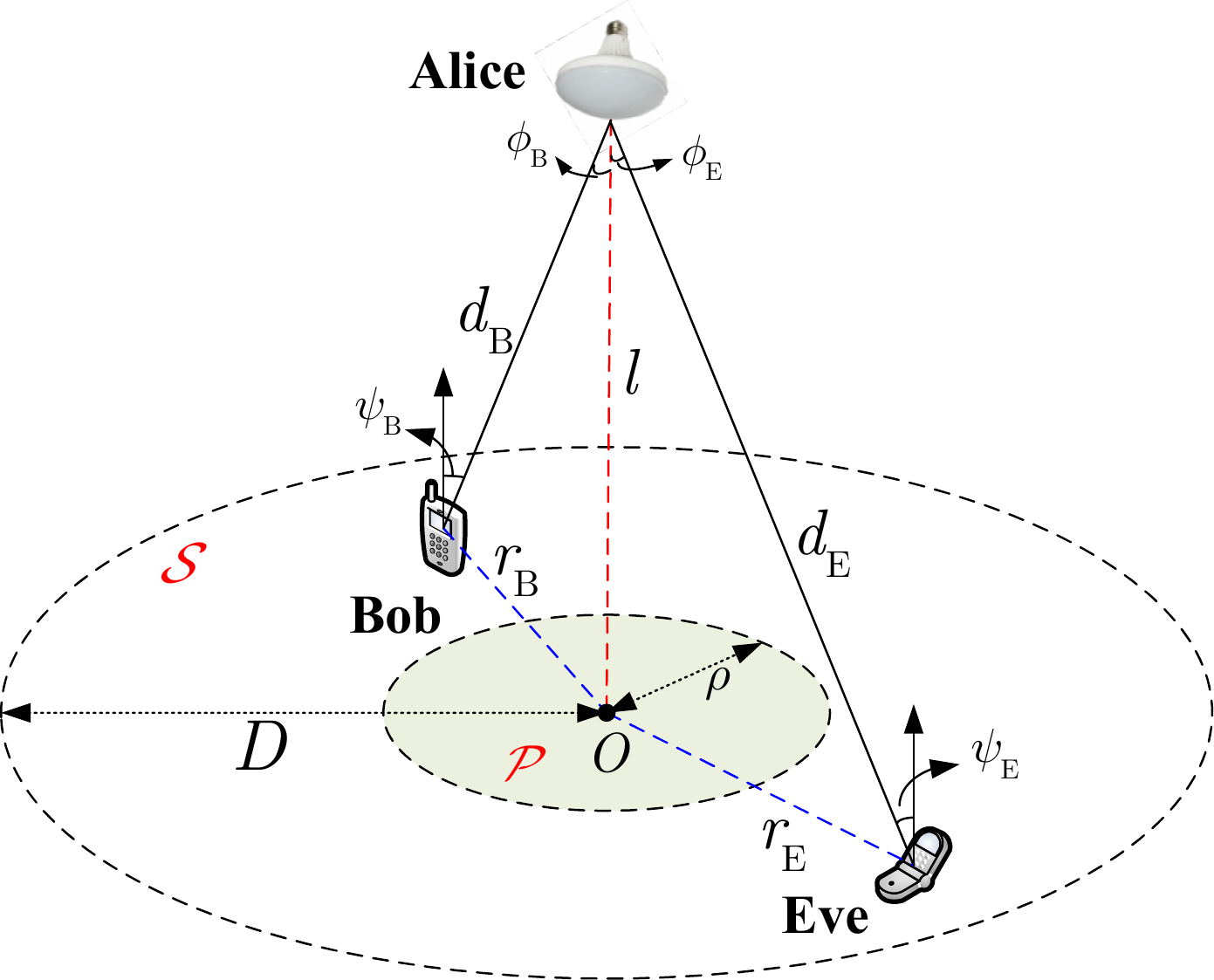}
\caption{An indoor VLC network with a protected zone.}
\label{fig1}
\end{figure}

To facilitate the analysis, the receiver zone ${\cal S}$ is assumed to be a disc with radius $D$,
and the projection of Alice on the receiver zone is the center of the disc (i.e., point $O$), as shown in Fig. \ref{fig1}.
To improve the secure performance, a protected zone ${\cal P}$ is employed \cite{BIB21}, which is also a disc with center $O$ and radius $\rho $.
The protected zone is an eavesdropper-free area, which is formed either inherently or intentionally \cite{BIB22}.
Therefore, we assume that Bob is uniformly distributed in zone ${\cal S}$, and Eve is uniformly distributed in zone ${\cal S}\backslash {\cal P}$. According to the above setup, the probability density functions (PDFs) of the positions of Bob and Eve are, respectively, given by
\begin{equation}
\left\{ \begin{array}{l}
{f_U}(u) = \frac{1}{{\pi {D^2}}},\;u \in {\cal S}\\
{f_W}(w) = \frac{1}{{\pi ({D^2} - {\rho ^2})}},\;w \in {\cal S}\backslash {\cal P}
\end{array} \right.,
\label{eq1}
\end{equation}
where $U$ and $W$ denote the positions of Bob and Eve.

Without loss of generality, we assume that both Bob and Eve can be illuminated by the LED, and thus the incidence angle ${\psi _k}$ ($k = {\rm{B}}$ for Bob, and $k = {\rm{E}}$ for Eve) cannot exceed the field of view of the PD ${\Psi _{\rm{c}}}$, i.e., $0 \le {\psi _k} \le {\Psi _{\rm{c}}}$. In indoor VLC, the channel gain ${H_k}$ can be written as \cite{BIB03}
\begin{equation}
{H_k} = \frac{{(m + 1)A}}{{2\pi d_k^2}}{T_s}g{\cos ^m}({\phi _k})\cos ({\psi _k}),\;k = {\rm{B}}\;{\rm{or}}\;{\rm{E,}}
\label{eq2}
\end{equation}
where ${d_k}$ and ${\phi _k}$ denote the distance and the irradiance angle between Alice and Bob (or Eve); $A$ denotes the physical area of the PD, $m$ is the order of the Lambertian emission,
$g$ is the concentrator gain of the PD, $T_s$ is the optical filter gain.

Moreover, the normal vectors of the transceiver planes are supposed to be perpendicular to the ceiling, and thus $\cos ({\phi _k}){\rm{ = }}\cos ({\psi _k}){\rm{ = }}l/{d_k}$, where $l$ is the vertical height between Alice and the receiver plane. Therefore, the channel gain in (\ref{eq2}) can be further written as
\begin{equation}
{H_k} = \frac{{(m + 1)A{T_s}g{l^{m{\rm{ + 1}}}}}}{{2\pi }}{(r_k^{\rm{2}} + {l^2})^{ - \frac{{m + 3}}{2}}},\;k = {\rm{B}}\;{\rm{or}}\;{\rm{E,}}
\label{eq3}
\end{equation}
where ${r_k}$ is the distance between the projection point $O$ and the $k$-th receiver, as shown in Fig. \ref{fig1}.

Owing to the channel estimation error, the CSIR is imperfect. According to the ellipsoidal approximation, the CSI uncertainty can be deterministically modeled as \cite{BIB23}
\begin{equation}
{\hat H_k} = {10^{\frac{{{\eta _k}}}{{10}}}}{H_k},\;\left| {{\eta _k}} \right| \le \varepsilon ,k = {\rm{B}}\;{\rm{or}}\;{\rm{E,}}
\label{eq4}
\end{equation}
where ${\eta _k}$ is the uncertain parameter, $\varepsilon$ is the uncertainty bound.

Therefore, the received signals at Bob and Eve can be written as
\begin{equation}
\left\{ \begin{array}{l}
{Y_{\rm{B}}} = {{\hat H}_{\rm{B}}}X + {Z_{\rm{B}}}\\
{Y_{\rm{E}}} = {{\hat H}_{\rm{E}}}X + {Z_{\rm{E}}}
\end{array} \right.,
\label{eq5}
\end{equation}
where $X$ denotes the input optical intensity signal; ${Z_{\rm{B}}} \sim N(0,\sigma _{\rm{B}}^{\rm{2}})$ and ${Z_{\rm{E}}} \sim N(0,\sigma _{\rm{E}}^{\rm{2}})$ are the additive white Gaussian noises at Bob and Eve, and $\sigma _{\rm{B}}^2$ and $\sigma _{\rm{E}}^2$ are the corresponding noise variances.

In VLC, the intensity modulation and direct detection (IM/DD) are often employed, and thus $X$ in (\ref{eq5}) must satisfy the non-negative constraint, i.e.,
\begin{equation}
X \ge 0.
\label{eq6}
\end{equation}

To satisfy the illumination requirement, the average optical intensity cannot change with time. Therefore, the dimmable average optical intensity constraint is expressed as
\begin{equation}
E(X){\rm{ = }}\xi P,
\label{eq7}
\end{equation}
where $\xi  \in (0,1]$ is the dimming target, $P$ is the nominal optical intensity of the LED.

\section{Channel Characteristics Analysis}
\label{section3}
In this paper, Bob and Eve are spatially random deployed receivers. The stochastic characteristics of the channel gains for both the main channel and the eavesdropping channel will be analyzed in this section.

According to Fig. \ref{fig1}, eqs. (\ref{eq1}) and (\ref{eq3}), the PDFs of ${H_{\rm{B}}}$ and ${H_{\rm{E}}}$ are derived in the following theorem.

\begin{theorem}
When Bob is uniformly distributed in zone ${\cal S}$ and Eve is uniformly distributed in zone ${\cal S}\backslash {\cal P}$, the PDFs of ${H_{\rm{B}}}$ and ${H_{\rm{E}}}$ are given, respectively, by
\begin{equation}
{f_{{H_{\rm{B}}}}}(h){\kern 1pt} {\rm{ = }}{\Xi _1}{h^{ - \frac{2}{{m + 3}} - 1}},\;\;{v_{\min }} \le h \le {v_{\max }},
\label{eq8}
\end{equation}
and
\begin{equation}
{f_{{H_{\rm{E}}}}}(h){\kern 1pt} {\rm{ = }}{\Xi _{\rm{2}}}{h^{ - \frac{2}{{m + 3}} - 1}},\;\;{v_{\min }} \le h \le {v_{{\mathop{\rm m}\nolimits} {\rm{id}}}},
\label{eq9}
\end{equation}
where ${\Xi _{\rm{1}}}$, ${\Xi _{\rm{2}}}$, ${v_{\min }}$, ${v_{\rm mid }}$ and ${v_{\max }}$ are expressed as
\begin{equation}
{\Xi _{\rm{1}}} = \frac{2}{{(m + 3){D^2}}}{\left( {\frac{{(m + 1)A{T_s}g{l^{m{\rm{ + 1}}}}}}{{2\pi }}} \right)^{\frac{2}{{m + 3}}}},
\label{eq10}
\end{equation}

\begin{equation}
{\Xi _{\rm{2}}} = \frac{2}{{(m + 3)({D^2} - {\rho ^2})}}{\left( {\frac{{(m + 1)A{T_s}g{l^{m{\rm{ + 1}}}}}}{{2\pi }}} \right)^{\frac{2}{{m + 3}}}},
\label{eq11}
\end{equation}

\begin{equation}
{v_{\min }} = \frac{{(m + 1)A{T_s}g{l^{m{\rm{ + 1}}}}}}{{2\pi }}{({D^2} + {l^2})^{ - \frac{{m + 3}}{2}}},
\label{eq12}
\end{equation}

\begin{equation}
{v_{{\mathop{\rm m}\nolimits} {\rm{id}}}}{\rm{ = }}\frac{{(m + 1)A{T_s}g{l^{m{\rm{ + 1}}}}}}{{2\pi }}{({\rho ^2} + {l^2})^{ - \frac{{m + 3}}{2}}},
\label{eq13}
\end{equation}
and
\begin{equation}
{v_{\max }}{\rm{ = }}\frac{{(m + 1)A{T_s}g{l^{m{\rm{ + 1}}}}}}{{2\pi }}{l^{ - (m + 3)}}.
\label{eq14}
\end{equation}
\label{them1}
\end{theorem}

\begin{proof}
See Appendix \ref{appa}.
\end{proof}

\begin{remark}
When Bob is uniformly deployed in zone ${\cal S}$ and Eve is uniformly distributed in zone ${\cal S}\backslash {\cal P}$, both (\ref{eq8}) and (\ref{eq9}) are monotonically decreasing functions with respect to $h$.
When $\rho {\rm{ = 0}}$ in (\ref{eq9}), the protected zone disappears, the distribution of ${H_{\rm{E}}}$ is the same as that of ${H_{\rm{B}}}$. Moreover, with the increase of $\rho $, the value of ${v_{{\rm{mid}}}}$ in (\ref{eq9}) decreases,
and thus the curve of ${f_{{H_{\rm{E}}}}}(h){\kern 1pt}$ is compressed to the left.
\label{rem1}
\end{remark}

To justify \emph{Remark \ref{rem1}}, Fig. \ref{fig2} shows the comparisons between the PDF of ${H_{\rm{B}}}$ and the PDF of ${H_{\rm{E}}}$ when $m = {T_s} = g = 1$, $A{\rm{ = 1c}}{{\rm{m}}^2}$, $l = 3{\rm{m}}$ and $D = 5{\rm{m}}$.
According to this figure, it can be observed that the above remark is verified.

\begin{figure}
\centering
\includegraphics[width=7cm]{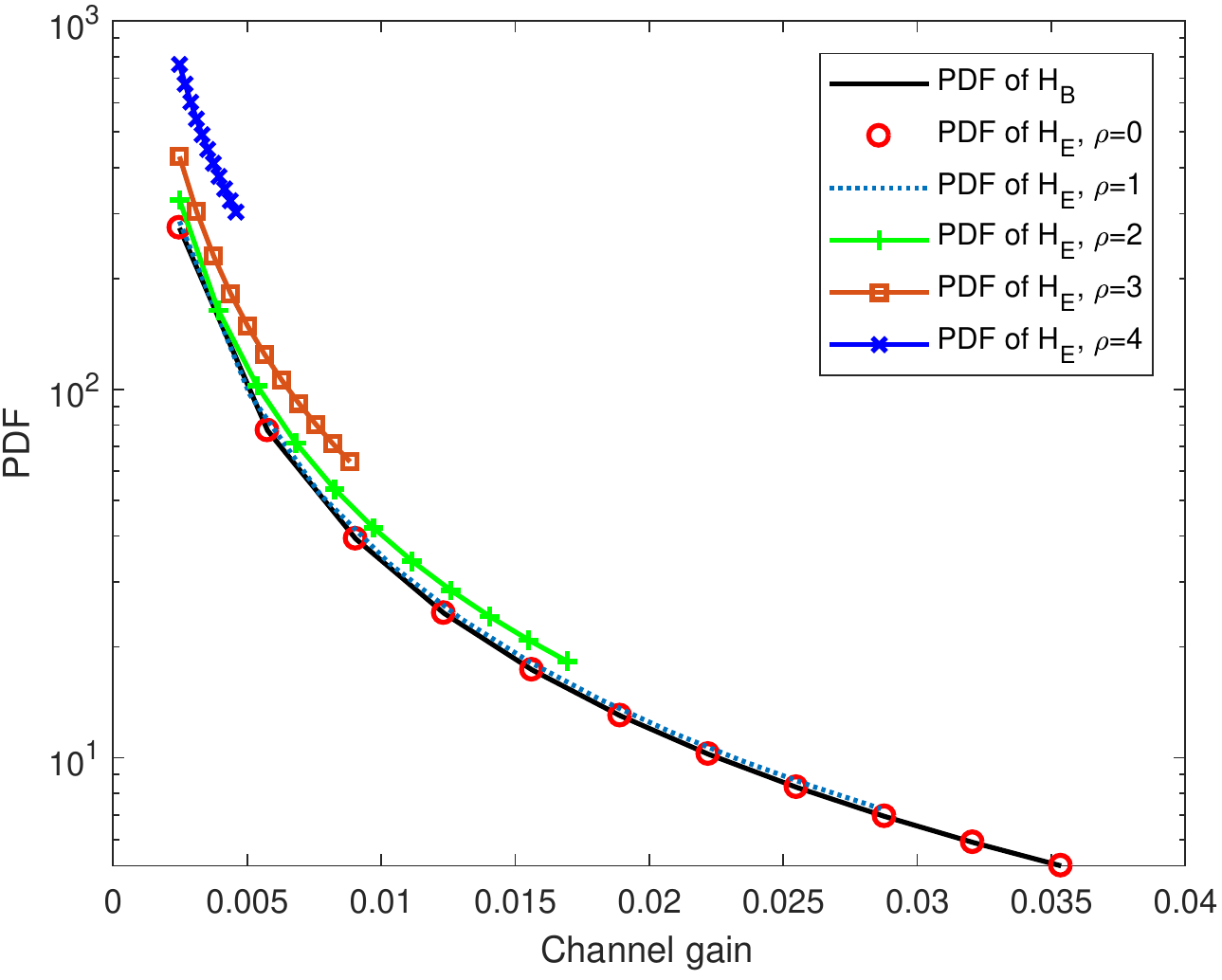}
\caption{Comparisons between the PDF of ${H_{\rm{B}}}$ and the PDF of ${H_{\rm{E}}}$ when $m = {T_s} = g = 1$, $A{\rm{ = 1c}}{{\rm{m}}^2}$, $l = 3{\rm{m}}$ and $D = 5{\rm{m}}$.}
\label{fig2}
\end{figure}

\section{Average Secrecy Capacity Analysis}
\label{section4}
In this section, the ASC performance for the random VLC network will be analyzed. By considering the imperfect CSIR, the randomly deployed receivers and protected zone for VLC, the instantaneous secrecy capacity (SC) cannot be used to evaluate the system performance. Based on the derived instantaneous SC result in our previous work \cite{BIB16}, the ASC will be analyzed to evaluate the secrecy performance. Moreover, some insights will also be provided.

In our previous work \cite{BIB16}, tight lower and upper bounds on the instantaneous SC for indoor VLC has been derived.
Without loss of generality, a lower bound on the instantaneous SC in \cite{BIB16} is employed to analyze the ASC, which is shown in the following Lemma.

\begin{lemma}
For VLC with constraints (\ref{eq6}) and (\ref{eq7}), a lower bound on the instantaneous SC is given by
\begin{equation}
{C_{\rm{s}}} \!\!=\!\! \left\{ \begin{array}{l}\!\!
\frac{1}{2}\ln \left( {\frac{{\sigma _{\rm{E}}^2}}{{2\pi \sigma _{\rm{B}}^2}} \cdot \frac{{e{\xi ^2}{P^2}H_{\rm{B}}^2 + 2\pi \sigma _{\rm{B}}^2}}{{H_{\rm{E}}^2{\xi ^2}{P^2} + \sigma _{\rm{E}}^2}}} \right),\;{\rm{if}}\;\chi '{H_{\rm{B}}} \ge {H_{\rm{E}}}\\
\!\!{\rm{0,}}\;{\rm{otherwise}}
\end{array} \right.\!.
\label{eq15}
\end{equation}
where $\chi ' \buildrel \Delta \over = \sqrt e {\sigma _{\rm{E}}}/(\sqrt {2\pi } {\sigma _{\rm{B}}})$.
\label{lem1}
\end{lemma}

In \emph{Lemma \ref{lem1}}, the lower bound (\ref{eq15}) is derived based on the perfect CSIR.
In this paper, when considering the imperfect CSIR, the lower bound (\ref{eq15}) should be modified as
\begin{equation}
{C_{\rm{s}}} = \left\{ \begin{array}{l}
C_s^ + ,\;{\rm{if}}\;\chi {H_{\rm{B}}} \ge {H_{\rm{E}}}\\
{\rm{0,}}\;\;\;{\rm{otherwise}}
\end{array} \right.,
\label{eq16}
\end{equation}
where $\chi  \buildrel \Delta \over = {\rm{1}}{{\rm{0}}^{\frac{{{\eta _{\rm{B}}} - {\eta _{\rm{E}}}}}{{10}}}}\sqrt e {\sigma _{\rm{E}}}/(\sqrt {2\pi } {\sigma _{\rm{B}}})$, and $C_s^ + $ is given by
\begin{equation}
C_s^ +  = \frac{1}{2}\ln \left( {\frac{{\sigma _{\rm{E}}^2}}{{2\pi \sigma _{\rm{B}}^2}} \cdot \frac{{e{\xi ^2}{P^2}{\rm{1}}{{\rm{0}}^{\frac{{{\eta _{\rm{B}}}}}{5}}}H_{\rm{B}}^2 + 2\pi \sigma _{\rm{B}}^2}}{{{\rm{1}}{{\rm{0}}^{\frac{{{\eta _{\rm{E}}}}}{5}}}H_{\rm{E}}^2{\xi ^2}{P^2} + \sigma _{\rm{E}}^2}}} \right).
\label{eq17}
\end{equation}

According to (\ref{eq16}), the ASC for the random VLC network is defined as
\begin{eqnarray}
\overline {{C_s}} % &=& E\left( {{C_s}} \right) \nonumber\\
 = \int_{{v_{\min }}}^{{v_{{\rm{mid}}}}} {\int_{{v_{\min }}}^{{v_{{\rm{max}}}}} {{C_s}  {f_{{H_{\rm{B}}}}}(x){f_{{H_{\rm{E}}}}}(y){\rm{d}}x{\rm{d}}y} } .
\label{eq18}
\end{eqnarray}
Note that (\ref{eq18}) holds because ${H_{\rm{B}}}$ and ${H_{\rm{E}}}$ are independent of each other.
In (\ref{eq18}), ${C_s}$ is positive for some cases,
while the value of ${C_s}$ is zero for other cases.
Moreover, the value of (\ref{eq18}) also depends on the integral region.
Fig. \ref{fig3} shows five cases of the integral region in (\ref{eq18}).
According to the five cases, the ACS in (\ref{eq18}) can be written as
\begin{equation}
\overline {{C_{\rm{s}}}} \!\!=\!\!\! \left\{ \begin{array}{l}\!\!\!\!
{\rm{0,}}\;\;\;\;\;\;\;\;\;\;\;\;\;\;\;\;\;\;\;\;\;\;\;\;\;\;\;\;\;\;\;\;\;\;\;\;\;\;\;\;\;\;\;\;\;\;\;\;\;\;\;\;\;{\rm{if}}\;\chi  < \frac{{{v_{\min }}}}{{{v_{\max }}}}\\
\!\!\!\!\int_{\frac{{{v_{\min }}}}{\chi }}^{{v_{{\rm{max}}}}}\! {{f_{{H_{\rm{B}}}}}\!(x)\!\int_{{v_{\min }}}^{\chi x}\! {C_{\rm{s}}^ + {f_{{H_{\rm{E}}}}}(y){\rm{d}}y} } {\kern 1pt} {\rm{d}}x,\;{\rm{if}}\;\frac{{{v_{\min }}}}{{{v_{\max }}}} \!\le\! \chi  \!<\! \frac{{{v_{{\mathop{\rm mid}\nolimits} }}}}{{{v_{\max }}}}\\
\!\!\!\!\int_{{v_{\min }}}^{{v_{{\rm{mid}}}}} \!{{f_{{H_{\rm{E}}}}}\!(y)\!\int_{\frac{y}{\chi }}^{{v_{\max }}}\! {C_{\rm{s}}^ + {f_{{H_{\rm{B}}}}}(x){\rm{d}}} } x{\kern 1pt} {\rm{d}}y,\;{\rm{if}}\;\frac{{{v_{{\mathop{\rm mid}\nolimits} }}}}{{{v_{\max }}}} \!\le\! \chi  \!<\! {\rm{1}}\\
\!\!\!\!\int_{{v_{\min }}}^{\frac{{{v_{{\mathop{\rm mid}\nolimits} }}}}{\chi }} \!{{f_{{H_{\rm{B}}}}}\!(x)\!\int_{{v_{\min }}}^{\chi x} \!{C_{\rm{s}}^ +\! {f_{{H_{\rm{E}}}}}\!(y){\rm{d}}y} } {\rm{d}}x\\
\!\!\!\! + \int_{\frac{{{v_{{\mathop{\rm m}\nolimits} {\rm{id}}}}}}{\chi }}^{{v_{{\rm{max}}}}} \!{{f_{{H_{\rm{B}}}}}\!(x)\!\int_{{v_{\min }}}^{{v_{{\rm{mid}}}}} \!{C_{\rm{s}}^ +\! {f_{{H_{\rm{E}}}}}\!(y){\rm{d}}y} } {\rm{d}}x,\;{\rm{if}}\;{\rm{1}} \!\le
 \!\chi \! <\! \frac{{{v_{{\mathop{\rm mid}\nolimits} }}}}{{{v_{\min }}}}\\
\!\!\!\!\int_{{v_{\min }}}^{{v_{{\rm{max}}}}} \!{{f_{{H_{\rm{B}}}}}\!(x)\!\int_{{v_{\min }}}^{{v_{{\rm{mid}}}}}\! {C_{\rm{s}}^ +\! {f_{{H_{\rm{E}}}}}\!(y)\!{\rm{d}}y} } {\rm{d}}x,\;{\rm{if}}\;\chi  \!\ge\! \frac{{{v_{{\mathop{\rm mid}\nolimits} }}}}{{{v_{\min }}}}
\end{array} \right.
\label{eq19}
\end{equation}

\begin{figure}
\centering
\includegraphics[width=8cm]{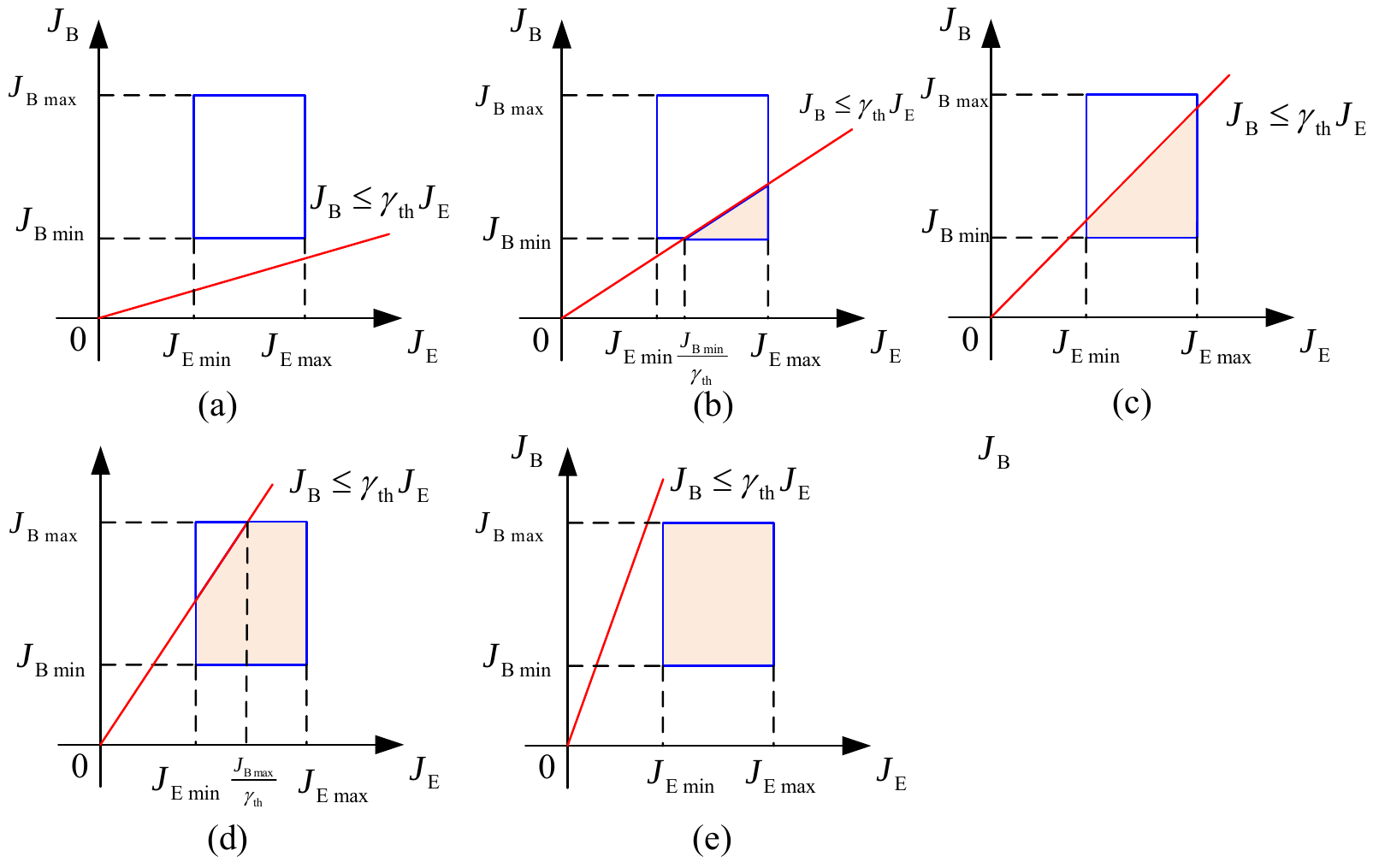}
\caption{Five cases of the integral region of (\ref{eq18}).}
\label{fig3}
\end{figure}

By solving (\ref{eq19}), the theoretical expression of the ASC is derived as the following theorem.

\begin{theorem}
By considering the imperfect CSIR and protected zone, the ASC for the random VLC network with constraints (\ref{eq6}) and (\ref{eq7}) are given by
\begin{equation}
\overline {{C_s}} {\rm{ = }}\left\{ \begin{array}{l}
{\rm{0,}}\;{\rm{if}}\;\chi  < \frac{{{v_{\min }}}}{{{v_{\max }}}}\;\\
{C_1},\;{\rm{if}}\;\frac{{{v_{\min }}}}{{{v_{\max }}}} < \chi  \le \frac{{{v_{{\mathop{\rm mid}\nolimits} }}}}{{{v_{\max }}}}\;\\
{C_2},\;{\rm{if}}\;\frac{{{v_{{\mathop{\rm mid}\nolimits} }}}}{{{v_{\max }}}} < \chi  \le 1\;\\
{C_3},\;{\rm{if}}\;1 < \chi  \le \frac{{{v_{{\mathop{\rm mid}\nolimits} }}}}{{{v_{\min }}}}\;\\
{C_4},\;{\rm{if}}\;\chi  > \frac{{{v_{{\mathop{\rm mid}\nolimits} }}}}{{{v_{\min }}}}
\end{array} \right.,
\label{eq20}
\end{equation}
where $C_1$ is given by
\begin{eqnarray}
{C_1} \!\!\!\!&=&\!\!\!\! \frac{{(m \!+\! 3){\Xi _{\rm{1}}}{\Xi _{\rm{2}}}}}{4}\left[\! {v_{\min }^{ - \frac{2}{{m + 3}}}\lambda\! \left( {\frac{{m \!+\! 3}}{2},\frac{{{v_{\min }}}}{\chi },{v_{\max }},\frac{{e{\rm{1}}{{\rm{0}}^{\frac{{{\eta _{\rm{B}}}}}{5}}}}}{{2\pi \sigma _{\rm{B}}^2}}} \right)} \right. \nonumber\\
&-&\!\!\!\! {\chi ^{ - \frac{2}{{m + 3}}}}\lambda\! \left( {\frac{{m \!+\! 3}}{{\rm{4}}},\frac{{{v_{\min }}}}{\chi },{v_{\max }},\frac{{e{\rm{1}}{{\rm{0}}^{\frac{{{\eta _{\rm{B}}}}}{5}}}}}{{2\pi \sigma _{\rm{B}}^2}}} \right) \nonumber\\
 &-&\!\!\!\! {\chi ^{\frac{2}{{m + 3}}}}\lambda \left( {\frac{{m + 3}}{{\rm{4}}},{v_{\min }},\chi {v_{\max }},\frac{{{\rm{1}}{{\rm{0}}^{\frac{{{\eta _{\rm{E}}}}}{5}}}}}{{\sigma _{\rm{E}}^2}}} \right)\nonumber\\
 &+&\left. {  v_{{\rm{max}}}^{ - \frac{2}{{m + 3}}}\lambda \left( {\frac{{m + 3}}{{\rm{2}}},{v_{\min }},\chi {v_{\max }},\frac{{{\rm{1}}{{\rm{0}}^{\frac{{{\eta _{\rm{E}}}}}{5}}}}}{{\sigma _{\rm{E}}^2}}} \right)} \right],
 \label{eq21}
\end{eqnarray}
Moreover, $C_2$ and $C_3$ are given by (\ref{eq22}) and (\ref{eq23}) as shown at the top of the next page.
\begin{table*}\normalsize
\begin{eqnarray}
{C_{\rm{2}}} \!\!\!\!\!\!&=&\!\!\!\!\!\! \frac{{(m \!+\! 3){\Xi _{\rm{1}}}{\Xi _{\rm{2}}}}}{4}\!\!\left[\! {v_{\min }^{ - \frac{2}{{m + 3}}}\lambda\! \left(\!\! {\frac{{m \!+\! 3}}{2},\frac{{{v_{\min }}}}{\chi },\frac{{{v_{{\rm{mid}}}}}}{\chi },\frac{{e{\rm{1}}{{\rm{0}}^{\frac{{{\eta _{\rm{B}}}}}{5}}}}}{{2\pi \sigma _{\rm{B}}^2}}} \right) \!-\! {\chi ^{ - \frac{2}{{m + 3}}}}\lambda\! \left( {\frac{{m \!+\! 3}}{4},\frac{{{v_{\min }}}}{\chi },\frac{{{v_{{\rm{mid}}}}}}{\chi },\frac{{e{\rm{1}}{{\rm{0}}^{\frac{{{\eta _{\rm{B}}}}}{5}}}}}{{2\pi \sigma _{\rm{B}}^2}}} \right)} +\left( {v_{\min }^{ - \frac{2}{{m + 3}}} - v_{{\mathop{\rm mi}\nolimits} {\rm{d}}}^{ - \frac{2}{{m + 3}}}} \right)\right.\nonumber \\
 &\times&\!\!\!\!\!\! \lambda \left(\! {\frac{{m \!+\! 3}}{2},\frac{{{v_{{\rm{mid}}}}}}{\chi },{v_{\max }},\frac{{e{\rm{1}}{{\rm{0}}^{\frac{{{\eta _{\rm{B}}}}}{5}}}}}{{2\pi \sigma _{\rm{B}}^2}}} \!\right)\!-\! \left. {  {\chi ^{\frac{2}{{m + 3}}}}\lambda \left(\! {\frac{{m \!+\! 3}}{4},{v_{\min }},{v_{{\mathop{\rm mid}\nolimits} }},\frac{{{\rm{1}}{{\rm{0}}^{\frac{{{\eta _{\rm{E}}}}}{5}}}}}{{\sigma _{\rm{E}}^2}}} \!\right) \!+\! v_{{\rm{max}}}^{ - \frac{2}{{m + 3}}}\lambda \left(\! {\frac{{m \!+\! 3}}{2},{v_{\min }},{v_{{\mathop{\rm mid}\nolimits} }},\frac{{{\rm{1}}{{\rm{0}}^{\frac{{{\eta _{\rm{E}}}}}{5}}}}}{{\sigma _{\rm{E}}^2}}} \!\right)} \!\right],
 \label{eq22}
\end{eqnarray}
\hrulefill
\end{table*}
\begin{table*}\normalsize
\begin{eqnarray}
{C_{\rm{3}}}  \!\!\!\!&=&\!\!\!\! \frac{{m \!+\! 3}}{4}{\Xi _{\rm{1}}}{\Xi _{\rm{2}}}\!\!\left\{\! {v_{\min }^{ - \frac{2}{{m + 3}}}\lambda\! \left(\!\! {\frac{{m \!+\! 3}}{2},{v_{\min }},\frac{{{v_{{\mathop{\rm mid}\nolimits} }}}}{\chi },\frac{{e{\rm{1}}{{\rm{0}}^{\frac{{{\eta _{\rm{B}}}}}{5}}}}}{{2\pi \sigma _{\rm{B}}^2}}} \right) \!-\! {\chi ^{ - \frac{2}{{m + 3}}}}\lambda\! \left(\!\! {\frac{{m \!+\! 3}}{4},{v_{\min }},\frac{{{v_{{\mathop{\rm mid}\nolimits} }}}}{\chi },\frac{{e{\rm{1}}{{\rm{0}}^{\frac{{{\eta _{\rm{B}}}}}{5}}}}}{{2\pi \sigma _{\rm{B}}^2}}} \right)}-\left[\! {v_{\min }^{ - \frac{2}{{m + 3}}} \!-\! {{\left(\! {\frac{{{v_{{\rm{mid}}}}}}{\chi }} \!\right)\!}^{ - \frac{2}{{m + 3}}}}} \right] \right. \nonumber \\
 &\times&\!\!\!\!\! \lambda\! \left(\!\! {\frac{{m \!+\! 3}}{2},{v_{\min }},\chi {v_{\min }},\frac{{{\rm{1}}{{\rm{0}}^{\frac{{{\eta _{\rm{E}}}}}{5}}}}}{{\sigma _{\rm{E}}^2}}}\! \right) \!-\! {\chi ^{\frac{2}{{m + 3}}}}\lambda\! \left(\!\! {\frac{{m \!+\! 3}}{4},\chi {v_{\min }},{v_{{\rm{mid}}}},\frac{{{\rm{1}}{{\rm{0}}^{\frac{{{\eta _{\rm{E}}}}}{5}}}}}{{\sigma _{\rm{E}}^2}}} \right)+{\left(\! {\frac{{{v_{{\rm{mid}}}}}}{\chi }} \!\right)^{\! - \frac{2}{{m + 3}}}}\lambda\! \left(\!\! {\frac{{m \!+\! 3}}{2},\chi {v_{\min }},{v_{{\rm{mid}}}},\frac{{{\rm{1}}{{\rm{0}}^{\frac{{{\eta _{\rm{E}}}}}{5}}}}}{{\sigma _{\rm{E}}^2}}} \!\right) \nonumber \\
 &+&\!\!\!\!\! \left(\! {v_{\min }^{ - \frac{2}{{m + 3}}} \!-\! v_{{\rm{mid}}}^{ - \frac{2}{{m + 3}}}}\! \right)\lambda \!\left(\!\! {\frac{{m \!+\! 3}}{2},\frac{{{v_{{\rm{mid}}}}}}{\chi },{v_{\max }},\frac{{e{\rm{1}}{{\rm{0}}^{\frac{{{\eta _{\rm{B}}}}}{5}}}}}{{2\pi \sigma _{\rm{B}}^2}}} \!\right)- \left. {  \left[ {{{\left( {\frac{{{v_{{\rm{mid}}}}}}{\chi }} \right)}^{ - \frac{2}{{m + 3}}}} - v_{\max }^{ - \frac{2}{{m + 3}}}} \right]\lambda \left( {\frac{{m + 3}}{2},{v_{\min }},{v_{{\rm{mid}}}},\frac{{{\rm{1}}{{\rm{0}}^{\frac{{{\eta _{\rm{E}}}}}{5}}}}}{{\sigma _{\rm{E}}^2}}} \right)} \right\},
 \label{eq23}
\end{eqnarray}
\hrulefill
\end{table*}
$C_4$ is given by
\begin{eqnarray}
{C_4}\!\!\!\!\!\!&=&\!\!\!\!\!\!\frac{{(m \!+\! 3){\Xi _{\rm{1}}}\!{\Xi _{\rm{2}}}}}{4}\!\!\left[\!\! {\left(\!\! {v_{\min }^{ - \frac{2}{{m \!+\! 3}}} \!\!-\!\! v_{{\rm{mid}}}^{ - \frac{2}{{m \!+\! 3}}}}\!\! \right)\!\lambda\!\! \left(\!\! {\frac{{m \!\!+\!\! 3}}{2},\!{v_{\min }},\!{v_{\max }},\!\frac{{e{\rm{1}}{{\rm{0}}^{\frac{{{\eta _{\rm{B}}}}}{5}}}}}{{2\pi \!\sigma _{\rm{B}}^2}}} \!\!\right)} \right.\nonumber \\
&-&\!\!\!\!\!\!\left. {  \left(\!\! {v_{\min }^{ - \frac{2}{{m \!+\! 3}}} \!-\! v_{\max }^{ - \frac{2}{{m \!+\! 3}}}}\!\! \right)\lambda\!\! \left(\!\! {\frac{{m \!\!+\!\! 3}}{2},\!{v_{\min }},\!{v_{{\rm{mid}}}},\!\frac{{{\rm{1}}{{\rm{0}}^{\frac{{{\eta _{\rm{E}}}}}{5}}}}}{{\sigma _{\rm{E}}^2}}}\!\! \right)}\!\! \right],
 \label{eq24}
\end{eqnarray}
where $\lambda (a,b,c,d)$ in (\ref{eq21})-(\ref{eq24}) is defined as
\begin{eqnarray}
\lambda (a,b,c,d) \!\!\!\!&=&\!\!\!\! a\left( {{b^{ - \frac{1}{a}}} - {c^{ - \frac{1}{a}}}} \right)\ln \left( {2\pi \sigma _{\rm{B}}^2\sigma _{\rm{E}}^2} \right) \nonumber\\
&+&\!\!\!\! \frac{1}{2}{c^{ - \frac{1}{a}}}G_{3,3}^{1,3}\left[ {d{\xi ^2}{P^2}{c^2}\left| \begin{array}{l}
1,1,1{\rm{ + }}\frac{1}{{2a}}\\
1,\frac{1}{{2a}},0
\end{array} \right.} \right] \nonumber\\
&-&\!\!\!\! \frac{1}{2}{b^{ - \frac{1}{a}}}G_{3,3}^{1,3}\left[ {d{\xi ^2}{P^2}{b^2}\left| \begin{array}{l}
1,1,1{\rm{ + }}\frac{1}{{2a}}\\
1,\frac{1}{{2a}},0
\end{array} \right.} \right],
 \label{eq25}
\end{eqnarray}
and $G_{a,b}^{c,d}[\cdot]$ denotes the Merjer's G function \cite{BIB23_1}.
\label{them2}
\end{theorem}

\begin{proof}
See Appendix \ref{appb}.
\end{proof}

\begin{remark}
From \emph{Theorem \ref{them2}}, with the increase of the uncertain parameter ${\eta _{\rm{B}}}$, $\chi$ increases, and the integral range of the ASC also enlarges. Therefore, the ASC will increase with the increase of ${\eta _{\rm{B}}}$. This indicates that the larger the uncertain parameter of the main channel is, the better the secure performance becomes. Similarly, we can also conclude that the ASC performance degrades with the increase of ${\eta _{\rm{E}}}$.
\label{rem2}
\end{remark}

\begin{remark}
With the increase of $P$, the integral range of the ASC does not change, but the instantaneous SC increases first and then tends to a stable value, and thus the ASC also increases first and then tends to a stable value.
\label{rem3}
\end{remark}

\begin{remark}
Intuitively, when the radius $\rho $ is larger, the feasible zone of Eve (i.e., ${\cal S}\backslash {\cal P}$) becomes smaller. In this case, the eavesdropping channel gets worse than the main channel. Therefore, with the increase of the radius $\rho $, the ASC performance improves. Oppositely, when $\rho  = 0$, Eve can be deployed in the same zone as Bob, the worst ASC achieves.
In a word, the secrecy performance improves by using the protected zone, however, the complexity of the system also increases. Therefore, a tradeoff between secrecy performance and system cost should be considered for practical systems.
\label{rem4}
\end{remark}

\section{Secrecy Outage Probability Analysis}
\label{section5}
In this section, the SOP for the random VLC network with imperfect CSI will be derived. It is very hard to obtain a closed-form expression for the exact SOP. Alternatively, a closed-form expression for a lower bound of the SOP is derived.

To facilitate the analysis, the instantaneous SC in (\ref{eq16}) can be further written as
\begin{eqnarray}
{C_{\rm{s}}} = \left\{ \begin{array}{l}
\frac{1}{2}\ln \left( {\frac{{1 + {J_{\rm{B}}}}}{{1 + {J_{\rm{E}}}}}} \right),\;{\rm{if}}{\kern 1pt} \chi {H_{\rm{B}}} \ge {H_{\rm{E}}}\\
{\rm{0,}}\;{\rm{otherwise}}
\end{array} \right.,
 \label{eq26}
\end{eqnarray}
where ${J_{\rm{B}}}$ and ${J_{\rm{E}}}$ are defined as
\begin{eqnarray}
\left\{ \begin{array}{l}
{J_{\rm{B}}} = \frac{{{\rm{1}}{{\rm{0}}^{\frac{{{\eta _{\rm{B}}}}}{5}}}e{\xi ^2}{P^2}H_{\rm{B}}^2}}{{2\pi \sigma _{\rm{B}}^2}}\\
{J_{\rm{E}}} = \frac{{{\rm{1}}{{\rm{0}}^{\frac{{{\eta _{\rm{E}}}}}{5}}}{\xi ^2}{P^2}H_{\rm{E}}^2}}{{\sigma _{\rm{E}}^2}}
\end{array} \right..
 \label{eq27}
\end{eqnarray}

According to (\ref{eq8}), the PDF of ${J_{\rm{B}}}$ is given by \cite{BIB24}
\begin{eqnarray}
{f_{{J_{\rm{B}}}}}(j) %&=& {f_{{H_{\rm{B}}}}}\left( {\sqrt {\frac{{2\pi \sigma _{\rm{B}}^2j}}{{e{\xi ^2}{P^2}{\rm{1}}{{\rm{0}}^{\frac{{{\eta _{\rm{B}}}}}{5}}}}}} } \right)\frac{{\rm{d}}}{{{\rm{d}}j}}\left( {\sqrt {\frac{{2\pi \sigma _{\rm{B}}^2j}}{{e{\xi ^2}{P^2}{\rm{1}}{{\rm{0}}^{\frac{{{\eta _{\rm{B}}}}}{5}}}}}} } \right)\nonumber \\
 \!\!\!\!\!&=&\!\!\!\!\! \frac{{{\Xi _1}}}{2}{\left( {\frac{{2\pi \sigma _{\rm{B}}^2}}{{e{\xi ^2}{P^2}{\rm{1}}{{\rm{0}}^{\frac{{{\eta _{\rm{B}}}}}{5}}}}}} \right)^{ - \frac{1}{{m + 3}}}}{j^{ - \frac{1}{{m + 3}} - 1}},\nonumber\\
 &&\;\frac{{e{\xi ^2}{P^2}{\rm{1}}{{\rm{0}}^{\frac{{{\eta _{\rm{B}}}}}{5}}}v_{\min }^2}}{{2\pi \sigma _{\rm{B}}^2}} \!\le\! j \!\le\! \frac{{e{\xi ^2}{P^2}{\rm{1}}{{\rm{0}}^{\frac{{{\eta _{\rm{B}}}}}{5}}}v_{\max }^2}}{{2\pi \sigma _{\rm{B}}^2}}.
 \label{eq28}
\end{eqnarray}

Similarly, according to (\ref{eq9}), the PDF of ${J_{\rm{E}}}$ is given by \cite{BIB24}
\begin{eqnarray}
{f_{{J_{\rm{E}}}}}(j) %&=& {f_{{H_{\rm{E}}}}}\left( {\sqrt {\frac{{\sigma _{\rm{E}}^2j}}{{{\rm{1}}{{\rm{0}}^{\frac{{{\eta _{\rm{E}}}}}{5}}}{\xi ^2}{P^2}}}} } \right){\kern 1pt} \frac{{\rm{d}}}{{{\rm{d}}j}}\left( {\sqrt {\frac{{\sigma _{\rm{E}}^2j}}{{{\rm{1}}{{\rm{0}}^{\frac{{{\eta _{\rm{E}}}}}{5}}}{\xi ^2}{P^2}}}} } \right) \nonumber \\
 \!\!\!\!&=&\!\!\!\! \frac{{{\Xi _{\rm{2}}}}}{2}{\left( {\frac{{\sigma _{\rm{E}}^2}}{{{\rm{1}}{{\rm{0}}^{\frac{{{\eta _{\rm{E}}}}}{5}}}{\xi ^2}{P^2}}}} \right)^{ - \frac{1}{{m + 3}}}}{j^{ - \frac{1}{{m + 3}} - 1}},\nonumber\\
 &&\frac{{{\rm{1}}{{\rm{0}}^{\frac{{{\eta _{\rm{E}}}}}{5}}}{\xi ^2}{P^2}v_{\min }^2}}{{\sigma _{\rm{E}}^2}} \le j \le \frac{{{\rm{1}}{{\rm{0}}^{\frac{{{\eta _{\rm{E}}}}}{5}}}{\xi ^2}{P^2}v_{{\rm{mid}}}^2}}{{\sigma _{\rm{E}}^2}}.
 \label{eq29}
\end{eqnarray}

Moreover, eq. (\ref{eq26}) can be further expressed as
\begin{eqnarray}
{C_{\rm{s}}} = \left\{ \begin{array}{l}
\frac{1}{2}\ln \left( {\frac{{1 + {J_{\rm{B}}}}}{{1 + {J_{\rm{E}}}}}} \right),\;{\rm{if}}{\kern 1pt} {J_{\rm{B}}} \ge {J_{\rm{E}}}\\
{\rm{0,}}\;{\rm{otherwise}}
\end{array} \right..
 \label{eq30}
\end{eqnarray}

According to (\ref{eq30}), the SOP is defined as
\begin{eqnarray}
{P_{{\rm{SOP}}}} %&=& \Pr \left( {{J_{\rm{B}}} \ge {J_{\rm{E}}}} \right)\Pr \left( {\left. {\frac{1}{2}\ln \left( {\frac{{1 + {J_{\rm{B}}}}}{{1 + {J_{\rm{E}}}}}} \right) \le \frac{1}{2}\ln {\gamma _{{\rm{th}}}}} \right|{J_{\rm{B}}} \ge {J_{\rm{E}}}} \right) \nonumber \\
% &+& \Pr \left( {{J_{\rm{B}}} < {J_{\rm{E}}}} \right)\underbrace {\Pr \left( {\left. {0 \le \frac{1}{2}\ln {\gamma _{{\rm{th}}}}} \right|{J_{\rm{B}}} < {J_{\rm{E}}}} \right)}_{ = 1} \nonumber\\
%&=& \Pr \left( {{J_{\rm{B}}} \ge {J_{\rm{E}}}} \right)\Pr \left( {\frac{1}{2}\ln \left( {\frac{{1 + {J_{\rm{B}}}}}{{1 + {J_{\rm{E}}}}}} \right) \le \frac{1}{2}\ln {\gamma _{{\rm{th}}}}|{J_{\rm{B}}} \ge {J_{\rm{E}}}} \right) \nonumber \\
% &+& \Pr \left( {{J_{\rm{B}}} < {J_{\rm{E}}}} \right)\Pr \left( {\left. {\frac{1}{2}\ln \left( {\frac{{1 + {J_{\rm{B}}}}}{{1 + {J_{\rm{E}}}}}} \right) \le \frac{1}{2}\ln {\gamma _{{\rm{th}}}}} \right|{J_{\rm{B}}} < {J_{\rm{E}}}} \right) \nonumber \\
 &=& \Pr \left( {\frac{1}{2}\ln \left( {\frac{{1 + {J_{\rm{B}}}}}{{1 + {J_{\rm{E}}}}}} \right) \le \frac{1}{2}\ln {\gamma _{{\rm{th}}}}} \right)\nonumber\\
 &=& \Pr \left( {{J_{\rm{B}}} \le (1 + {J_{\rm{E}}}){\gamma _{{\rm{th}}}} - 1} \right)\nonumber\\
 &\ge& \Pr \left( {{J_{\rm{B}}} \le {J_{\rm{E}}}{\gamma _{{\rm{th}}}}} \right) \buildrel \Delta \over = P_{{\rm{SOP}}}^{\rm{L}}
 \label{eq31}
\end{eqnarray}
where ${\gamma _{{\rm{th}}}} \ge {\rm{1}}$ denotes the equivalent threshold of the signal-to-noise ratio.

\begin{figure}
\centering
\includegraphics[width=8cm]{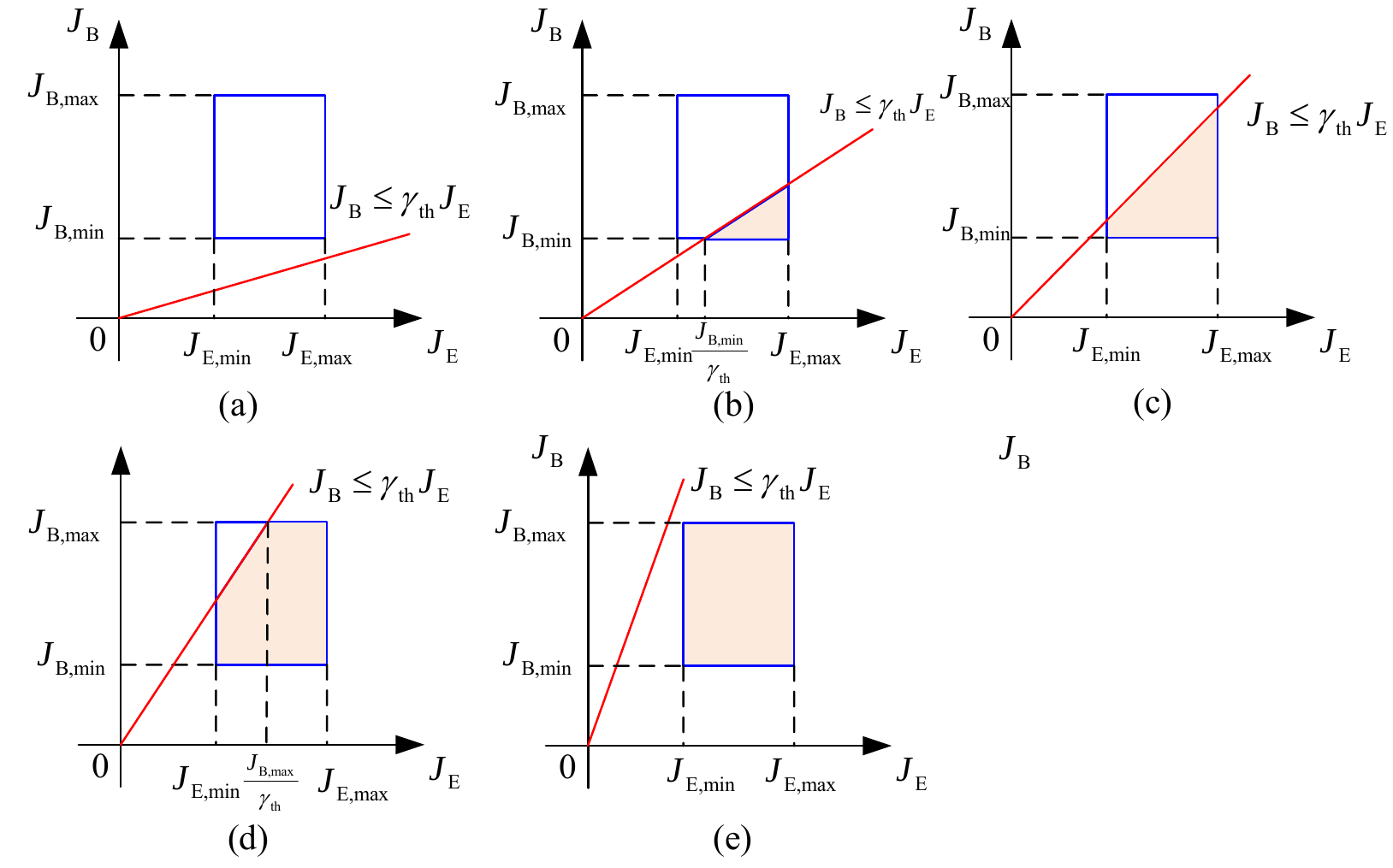}
\caption{Five cases of the integral region of (\ref{eq31}).}
\label{fig4}
\end{figure}

Note that the value of (\ref{eq31}) also depends on the integral region. Fig. \ref{fig4} shows five cases of the integral region in (\ref{eq31}). According to the five cases, the lower bound of SOP in (\ref{eq31}) can be written as
\begin{eqnarray}
P_{{\rm{SOP}}}^{\rm{L}}\!\!\!=\!\!\! \left\{ \begin{array}{l}\!\!\!\!\!
{\rm{0,}}\;{\rm{if}}\;{\gamma _{{\rm{th}}}} \le {\chi ^2}\frac{{v_{\min }^2}}{{v_{{\rm{mid}}}^2}}\;\\
\!\!\!\!\!\int_{\!\frac{{{J_{{\rm{B}},\min }}}}{{{\gamma _{{\rm{th}}}}}}}^{\!{J_{{\rm{E}},\max }}}\! \!{{f_{{J_{\rm{E}}}}}\!(\!z\!)\!\!\int_{\!{J_{{\rm{B}},\min }}}^{\!{\gamma _{{\rm{th}}}}z} \!\!{{f_{{J_{\rm{B}}}}}\!(\!y\!)\!{\rm{d}}y{\rm{d}}z} } ,{\rm{if}}{\chi ^2}\frac{{v_{\min }^2}}{{v_{{\rm{mid}}}^2}} \!\!<\!\! {\gamma _{{\rm{th}}}} \!\!\le\!\! {\chi ^2}\;\\
\!\!\!\!\!\int_{\!{J_{{\rm{E,min}}}}}^{\!{J_{{\rm{E,max}}}}} \!\!{{f_{{J_{\rm{E}}}}}\!(\!z\!)\!\!\int_{\!{J_{{\rm{B,min}}}}}^{\!{\gamma _{{\rm{th}}}}z} \!\!{{f_{{J_{\rm{B}}}}}\!(\!y\!)\!{\rm{d}}y{\rm{d}}z} },{\rm{if}}{\chi ^2} \!\!<\!\! {\gamma _{{\rm{th}}}} \!\!\le\!\! \frac{{\chi ^2}{v_{\max }^2}}{{v_{{\rm{mid}}}^2}}\\
\!\!\!\!\!\int_{\!{J_{{\rm{E,min}}}}}^{\!\frac{{{J_{{\rm{B,max}}}}}}{{{\gamma _{{\rm{th}}}}}}} \!\!{{f_{{J_{\rm{E}}}}}\!(\!z\!)\!\!\int_{\!{J_{{\rm{B,min}}}}}^{\!{\gamma _{{\rm{th}}}}z} \!\!{{f_{{J_{\rm{B}}}}}\!(\!y\!)\!{\rm{d}}y{\rm{d}}z} } +\\
\!\!\!\!\!  \int_{\!\frac{{{J_{{\rm{B,max}}}}}}{\!{{\gamma _{{\rm{th}}}}}}}^{{J_{{\rm{E,max}}}}} \!\!{{f_{{J_{\rm{E}}}}}\!(\!z\!)\!\!\int_{\!{J_{{\rm{B,min}}}}}^{\!{J_{{\rm{B,max}}}}}\!\! {{f_{{J_{\rm{B}}}}}\!(\!y\!)\!{\rm{d}}y{\rm{d}}z} } ,{\rm{if}}\frac{{\chi ^2}\!{v_{\max }^2}}{{v_{{\rm{mid}}}^2}} \!\!<\!\! {\gamma _{{\rm{th}}}} \!\!\le\!\! \frac{{\chi ^2}\!{v_{\max }^2}}{{v_{\min }^2}}\;\\
\!\!\!\!\!1,\;{\rm{if}}\;{\gamma _{{\rm{th}}}} \!>\! {\chi ^2}\frac{{v_{\max }^2}}{{v_{\min }^2}}
\end{array} \right.\!\!\!\!\!\!\!\!\!\!\!,
 \label{eq32}
\end{eqnarray}
where ${J_{{\rm{B,min}}}} = \frac{{e{\xi ^2}{P^2}{\rm{1}}{{\rm{0}}^{\frac{{{\eta _{\rm{B}}}}}{5}}}v_{\min }^2}}{{2\pi \sigma _{\rm{B}}^2}}$, ${J_{{\rm{B,max}}}} = \frac{{e{\xi ^2}{P^2}{\rm{1}}{{\rm{0}}^{\frac{{{\eta _{\rm{B}}}}}{5}}}v_{\max }^2}}{{2\pi \sigma _{\rm{B}}^2}}$, ${J_{{\rm{E,min}}}} = \frac{{{\rm{1}}{{\rm{0}}^{\frac{{{\eta _{\rm{E}}}}}{5}}}{\xi ^2}{P^2}v_{\min }^2}}{{\sigma _{\rm{E}}^2}}$, and ${J_{{\rm{E,max}}}} = \frac{{{\rm{1}}{{\rm{0}}^{\frac{{{\eta _{\rm{E}}}}}{5}}}{\xi ^2}{P^2}v_{{\rm{mid}}}^2}}{{\sigma _{\rm{E}}^2}}$.

By solving (\ref{eq32}), the closed-form expression of the lower bound of the SOP is obtained in the following theorem.

\begin{theorem}
By considering the imperfect CSIR and protected zone, the SOP for the random VLC network with constraints (\ref{eq6}) and (\ref{eq7}) are given by
\begin{eqnarray}
{P_{{\rm{SOP}}}}{\rm{ = }}\left\{ \begin{array}{l}
{\rm{0,}}\;{\rm{if}}\;{\gamma _{{\rm{th}}}} \le {\chi ^2}\frac{{v_{\min }^2}}{{v_{{\rm{mid}}}^2}}\;\\
{S_1},\;{\rm{if}}\;{\chi ^2}\frac{{v_{\min }^2}}{{v_{{\rm{mid}}}^2}} < {\gamma _{{\rm{th}}}} \le {\chi ^2}\;\\
{S_2},\;{\rm{if}}\;{\chi ^2}\; < {\gamma _{{\rm{th}}}} \le {\chi ^2}\;\frac{{v_{\max }^2}}{{v_{{\rm{mid}}}^2}}\\
{S_3},\;{\rm{if}}\;{\chi ^2}\;\frac{{v_{\max }^2}}{{v_{{\rm{mid}}}^2}} < {\gamma _{{\rm{th}}}} \le {\chi ^2}\frac{{v_{\max }^2}}{{v_{\min }^2}}\;\\
1,\;{\rm{if}}\;{\gamma _{{\rm{th}}}} > {\chi ^2}\frac{{v_{\max }^2}}{{v_{\min }^2}}
\end{array} \right.,
 \label{eq33}
\end{eqnarray}
where $S_1$, $S_2$ and $S_3$ are given by
\begin{eqnarray}
{S_1} \!\!\!\!&=&\!\!\!\! \psi {(m \!+\! 3)^2}J_{{\rm{B,min}}}^{ - \frac{1}{{m + 3}}}\!\!\left[\!\! {{\left(\!\! {\frac{{{J_{{\rm{B,min}}}}}}{{{\gamma _{{\rm{th}}}}}}}\!\! \right)\!}^{ - \frac{1}{{m + 3}}}}J_{{\rm{E,max}}}^{ - \frac{1}{{m + 3}}} \!\right] \nonumber\\
\!&-&\!\!\!\!  \frac{\psi }{2}{(m \!+\! 3)^2}\gamma _{{\rm{th}}}^{ - \frac{1}{{m + 3}}}\!\!\left[\!\! {{{\left(\!\! {\frac{{{J_{{\rm{B,min}}}}}}{{{\gamma _{{\rm{th}}}}}}}\!\! \right)\!}^{ - \frac{2}{{m + 3}}}} \!-\! J_{{\rm{E,max}}}^{ - \frac{2}{{m + 3}}}} \right],
 \label{eq34}
\end{eqnarray}

\begin{eqnarray}
{S_2} \!\!\!\!&=&\!\!\!\! \psi {(m + 3)^2}J_{{\rm{B,min}}}^{ - \frac{1}{{m + 3}}}\left( {J_{{\rm{E,min}}}^{ - \frac{1}{{m + 3}}} - J_{{\rm{E,max}}}^{ - \frac{1}{{m + 3}}}} \right) \nonumber\\
&-&\!\!\!\! \frac{\psi }{2}{(m + 3)^2}\gamma _{{\rm{th}}}^{ - \frac{1}{{m + 3}}}\left[ {J_{{\rm{E,min}}}^{ - \frac{2}{{m + 3}}} - J_{{\rm{E,max}}}^{ - \frac{2}{{m + 3}}}} \right],
 \label{eq35}
\end{eqnarray}
and
\begin{eqnarray}
\!\!\!\!\!\!\!\!\!\!\!\!\!\!\!{S_3}\!\!\!\!\! &=&\!\!\!\!\! \psi {(m \!+\! 3)^2}J_{{\rm{B,min}}}^{ - \frac{1}{{m + 3}}}\!\!\left[\! {J_{{\rm{E,min}}}^{ - \frac{1}{{m + 3}}} \!-\! {{\left(\!\! {\frac{{{J_{{\rm{B,max}}}}}}{{{\gamma _{{\rm{th}}}}}}} \!\!\right)\!}^{ - \frac{1}{{m + 3}}}}}\! \right] \nonumber\\
\!\!\!\!\!\!\!\!\!\!\!\!\!\!\!&-&\!\!\!\!\! \frac{\psi }{2}{(m \!+\! 3)^2}\gamma _{{\rm{th}}}^{ - \frac{1}{{m + 3}}}\!\!\left[\! {J_{{\rm{E,min}}}^{ - \frac{2}{{m + 3}}} \!-\! {{\left(\!\! {\frac{{{J_{{\rm{B,max}}}}}}{{{\gamma _{{\rm{th}}}}}}} \!\!\right)}^{ - \frac{2}{{m + 3}}}}}\! \right]\nonumber \\
\!\!\!\!\!\!\!\!\!\!\!\!\!\!\! &+&\!\!\!\!\!\! \psi {(m \!+\! 3)^2}\!\!\left(\!\! {J_{{\rm{B,min}}}^{ - \frac{1}{{m \!+\! 3}}} \!\!\!-\! J_{{\rm{B,max}}}^{ - \frac{1}{{m \!+\! 3}}}} \!\! \right)\!\!\!\left[\!\! {{{\left(\!\! {\frac{{{J_{{\rm{B,max}}}}}}{{{\gamma _{{\rm{th}}}}}}} \!\!\right)}^{\!\!\! - \frac{1}{{m \!+\! 3}}}} \!\!\!-\! J_{{\rm{E,max}}}^{ - \frac{1}{{m + 3}}}} \!\!\right]\!,
 \label{eq36}
\end{eqnarray}
where $\psi  = \frac{{{\Xi _1}{\Xi _{\rm{2}}}}}{4}{\left( {\frac{{2\pi \sigma _{\rm{B}}^2\sigma _{\rm{E}}^2}}{{{\rm{1}}{{\rm{0}}^{\frac{{{\eta _{\rm{E}}} + {\eta _{\rm{B}}}}}{5}}}e{\xi ^4}{P^4}}}} \right)^{ - \frac{1}{{m + 3}}}}$.
\label{them3}
\end{theorem}

\begin{proof}
See Appendix \ref{appc}.
\end{proof}

\begin{remark}
With the increase of ${\gamma _{{\rm{th}}}}$, the value of the lower bound of the SOP also enlarges. Finally, when ${\gamma _{{\rm{th}}}} > {{{\chi ^2}v_{\max }^2} \mathord{\left/ {\vphantom {{{\chi ^2}v_{\max }^2} {v_{\min }^2}}} \right. \kern-\nulldelimiterspace} {v_{\min }^2}}$, the lower bound of the SOP becomes one, and the information cannot be transmitted securely.
\label{rem5}
\end{remark}

\begin{remark}
Similar to \emph{Remark \ref{rem4}}, the SOP performance can also improve with the increase of $\rho$.
\label{rem6}
\end{remark}

\section{Numerical Results}
\label{section6}
In this section, the derived theoretical expressions of the ASC and the lower bound of the SOP will be verified by using Monte-Carlo simulations. The main simulation parameters are listed in Table \ref{tab1}.

\begin{table}[!h]
\caption{Main simulation parameters.}
\begin{center}
\begin{tabularx}{8.5cm}{|p{4cm}|X|X|}
\hline\hline
\centering Order of Lambertian emission &\centering $m$ &\centering 6
\tabularnewline\hline
\centering Physical area of PD &\centering ${A_r}$ &\centering 1 cm$^2$
\tabularnewline\hline
\centering Optical filter gain of PD &\centering ${T_s}$ &\centering 1
\tabularnewline\hline
\centering Concentrator gain of PD &\centering $g$ &\centering 3
\tabularnewline\hline
\centering Noise variances &\centering $\sigma _{\rm{B}}^2$, $\sigma _{\rm{E}}^2$ &\centering 1
\tabularnewline\hline\hline
\end{tabularx}
\end{center}
\label{tab1}
\end{table}

\subsection{ASC results}
\label{section6_1}
Fig. \ref{fig5} shows the ASC versus the nominal optical intensity $P$ with different $\rho$ when ${\eta _{\rm{B}}} = 10$, ${\eta _{\rm{E}}} = 1$, $l = 4{\rm{m}}$, and $D = 8{\rm{m}}$. As can be observed, when $P$ is small, all ASCs are zeros.
With the increase of $P$, the ASCs increase and then tend to stable values, which verifies \emph{Remark \ref{rem3}}.
Moreover, with the increase of the radius $\rho$ of the protected zone,
the ASC performance improves dramatically, which coincides with \emph{Remark \ref{rem4}}.
This is because Eve is located farther away from Bob when the protected zone is larger.

\begin{figure}
\centering
\includegraphics[width=8cm]{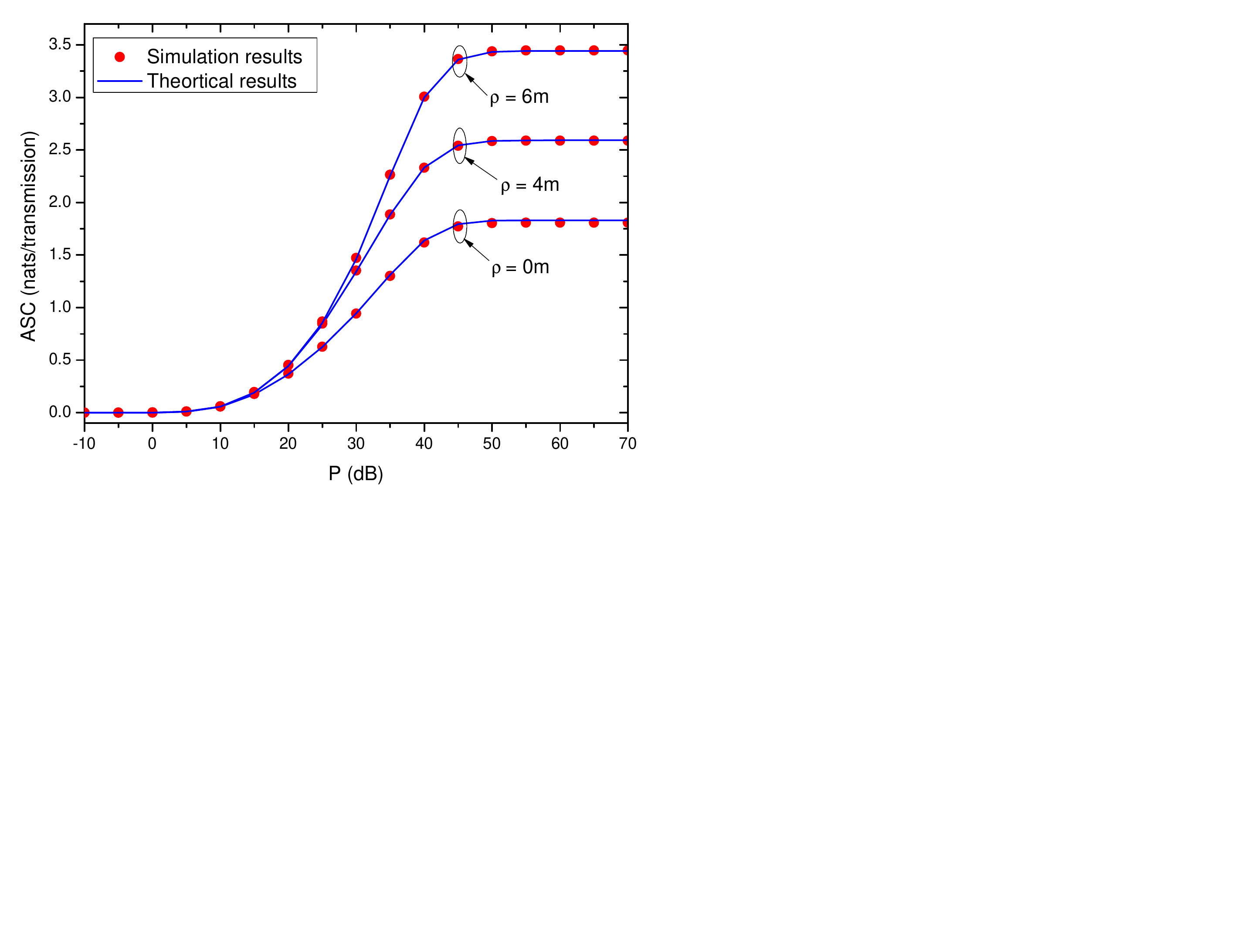}
\caption{ASC versus nominal optical intensity $P$ with different $\rho$ when ${\eta _{\rm{B}}} = 10$, ${\eta _{\rm{E}}} = 1$, $l = 4{\rm{m}}$, and $D = 8{\rm{m}}$.}
\label{fig5}
\end{figure}

Fig. \ref{fig6} plots the ASC versus the nominal optical intensity $P$ with different ${\eta _{\rm{B}}}$ when $\rho  = 4$, ${\eta _{\rm{E}}} = 1$, $l = 4{\rm{m}}$, and $D = 8{\rm{m}}$. According to (\ref{eq4}), it is known that the estimated channel gain ${\hat H_{\rm{B}}}$ increases with the increase of ${\eta _{\rm{B}}}$.
 Larger ${\hat H_{\rm{B}}}$ will lead to a good ASC performance.
 Therefore, the ASC performance improves with the increase of ${\eta _{\rm{B}}}$,
 which consists with that in \emph{Remark \ref{rem2}}.
 This indicates that the larger the uncertainty of the imperfect CSI, the larger the gap between the realistic ASC and the estimated ASC becomes.

\begin{figure}
\centering
\includegraphics[width=8cm]{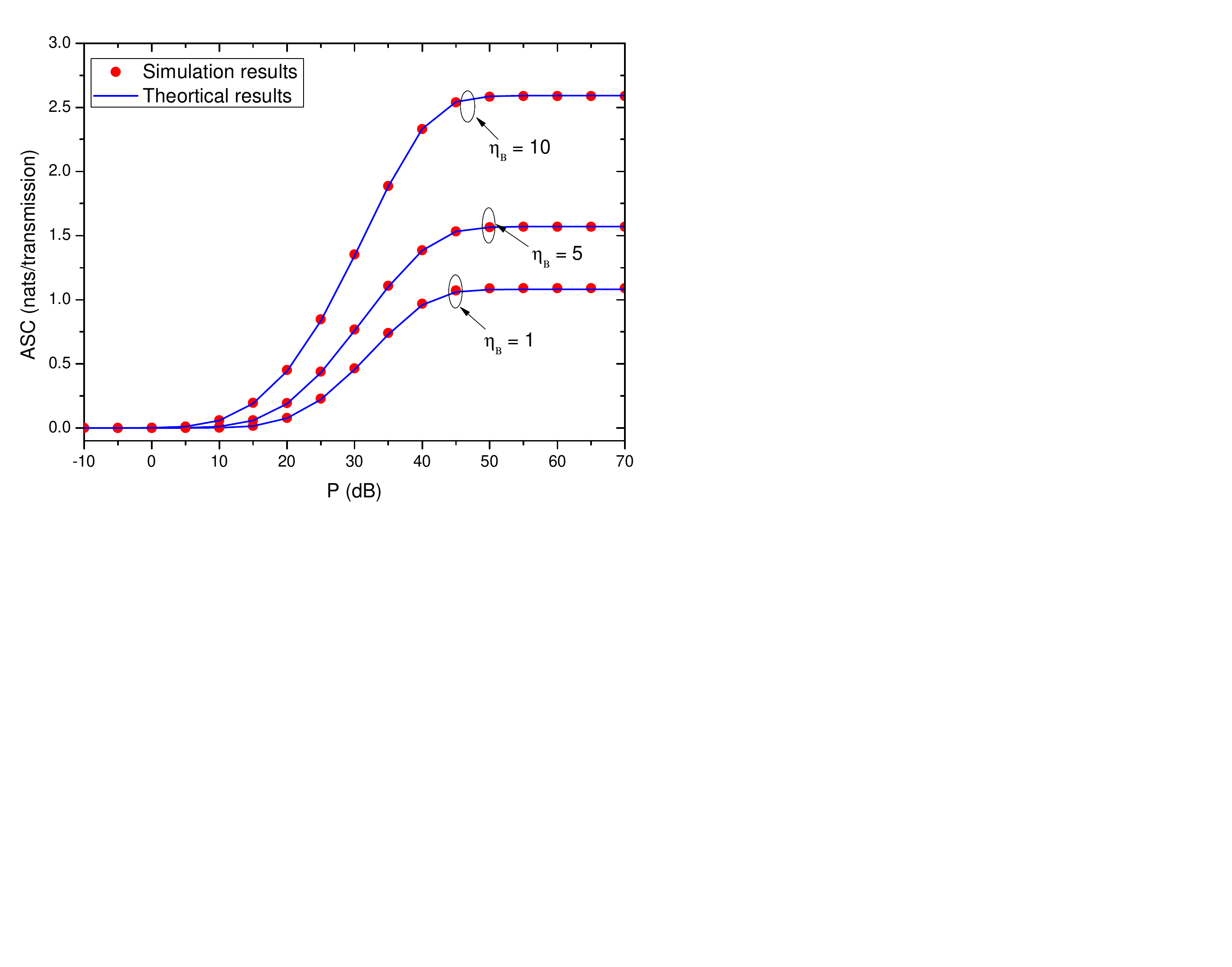}
\caption{ASC versus nominal optical intensity $P$ with different ${\eta _{\rm{B}}}$ when $\rho  = 4$, ${\eta _{\rm{E}}} = 1$, $l = 4{\rm{m}}$, and $D = 8{\rm{m}}$.}
\label{fig6}
\end{figure}

\begin{figure}
\centering
\includegraphics[width=8cm]{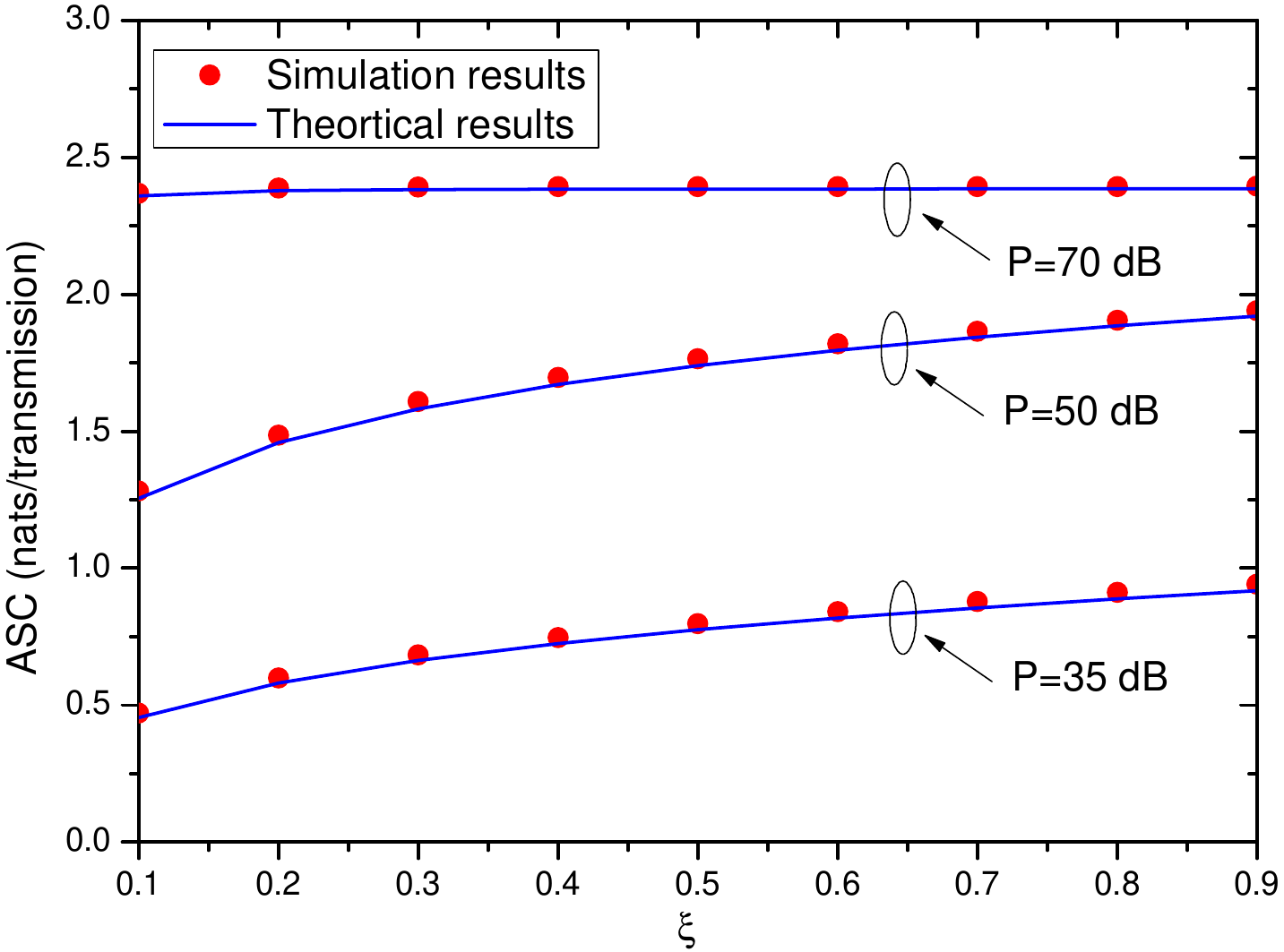}
\caption{ASC versus dimming target $\xi $ for different nominal optical intensity $P$ when $\rho  = 2{\kern 1pt} {\kern 1pt} {\rm{m}}$, ${\eta _{\rm{B}}} = 10$, ${\eta _{\rm{E}}} = 1$, $l = 4{\rm{m}}$, and $D = 15{\rm{m}}$.}
\label{fig7}
\end{figure}

Fig. \ref{fig7} shows ASC versus dimming target $\xi$ for different nominal optical intensity $P$ when $\rho  = 2{\kern 1pt} {\kern 1pt} {\rm{m}}$, ${\eta _{\rm{B}}} = 10$, ${\eta _{\rm{E}}} = 1$, $l = 4{\rm{m}}$, and $D = 15{\rm{m}}$. As can be observed, the ASC increase with the increase of $P$, which consists with that in Fig. \ref{fig5} and Fig. \ref{fig6}. Moreover, when $P = 35$  and 50 dB, the ASC varies fast with dimming target. Moreover, the larger the dimming target, the larger the ASC will be. However, when $P = 70$ dB, the values of ASC are almost the same for different dimming targets.

Moreover, it can be found from Fig. \ref{fig5}-Fig. \ref{fig7} that all theoretical results match simulation results very well, which indicates the correctness of the theoretical analysis.

\subsection{SOP results}
\label{section6_2}
Fig. \ref{fig8} shows SOP versus ${\eta _{\rm{B}}}$ with different ${\eta _{\rm{E}}}$ when $\rho  = 4{\rm{m}}$, $P=60$ dB, $l = 4{\rm{m}}$, $D = 8{\rm{m}}$, and ${\gamma _{{\rm{th}}}} = 3$.
As can be observed, the SOP performance improves with the increase of ${\eta _{\rm{B}}}$, and deteriorates with the increase of ${\eta _{\rm{E}}}$. This is because larger ${\eta _{\rm{B}}}$ and ${\eta _{\rm{E}}}$ will result in larger estimated channel gains ${\hat H_{\rm{B}}}$ and ${\hat H_{\rm{E}}}$, respectively. Larger ${\hat H_{\rm{B}}}$ leads to better secrecy performance while larger ${\hat H_{\rm{E}}}$ leads to worse secrecy performance.

\begin{figure}
\centering
\includegraphics[width=8cm]{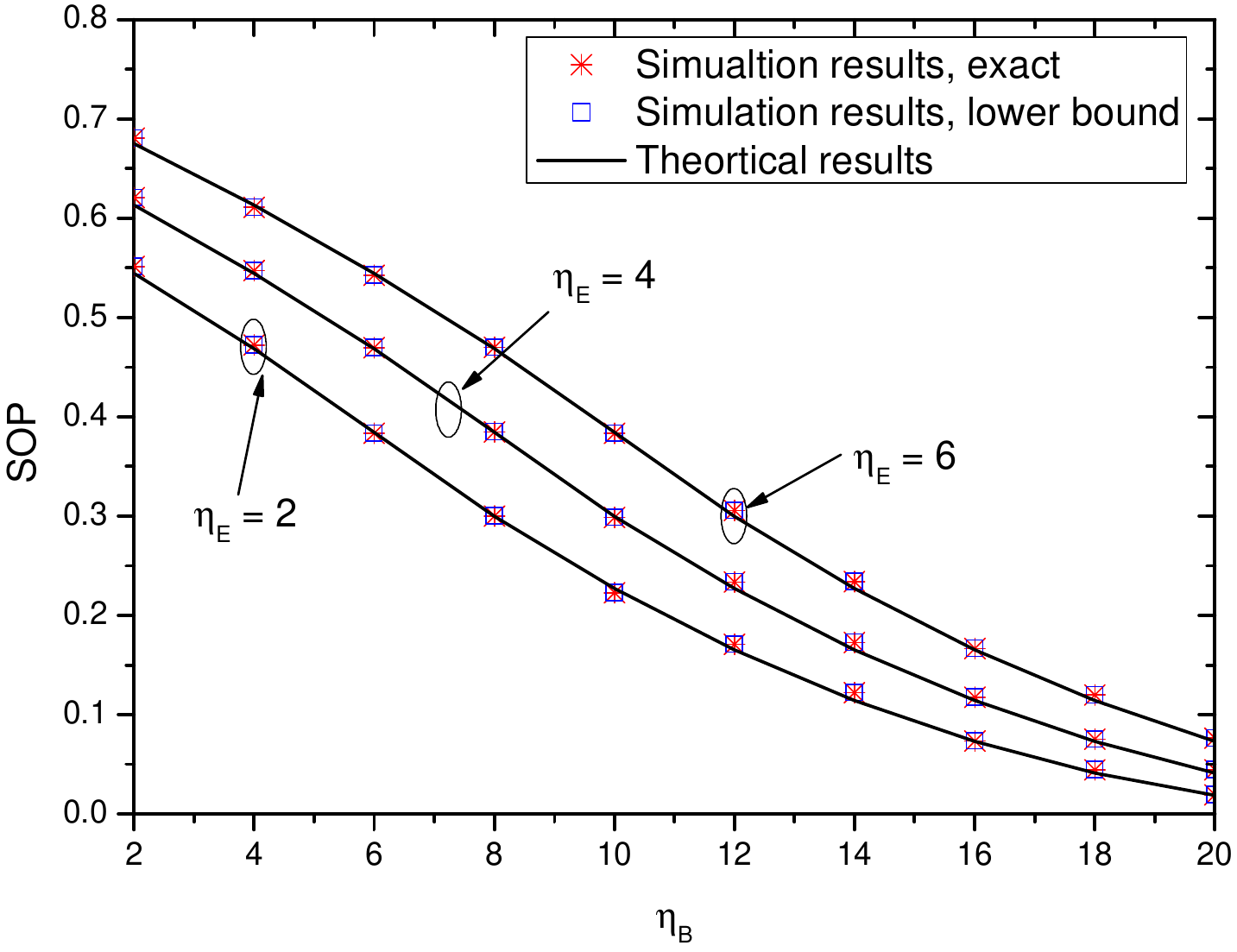}
\caption{SOP versus ${\eta _{\rm{B}}}$ with different ${\eta _{\rm{E}}}$ when $\rho  = 4{\rm{m}}$, $P=60$ dB, $l = 4{\rm{m}}$, $D = 8{\rm{m}}$, and ${\gamma _{{\rm{th}}}} = 3$.}
\label{fig8}
\end{figure}

Fig. \ref{fig9} shows SOP versus ${\eta _{\rm{B}}}$ with different $\rho$ when $P=60$ dB, $l=4{\rm m}$, $D = 8{\rm{m}}$, ${\gamma _{{\rm{th}}}} = 3$, and ${\eta _{\rm{E}}} = 1$.
It can be observed that the SOP performance improves with the increase of the radius of the protected zone $\rho$, which consists with that in \emph{Remark \ref{rem6}}. Expanding protected zone decreases the estimated eavesdropping channel gain ${\hat H_{\rm{E}}}$, which improves the secrecy performance of the VLC system.

\begin{figure}
\centering
\includegraphics[width=8cm]{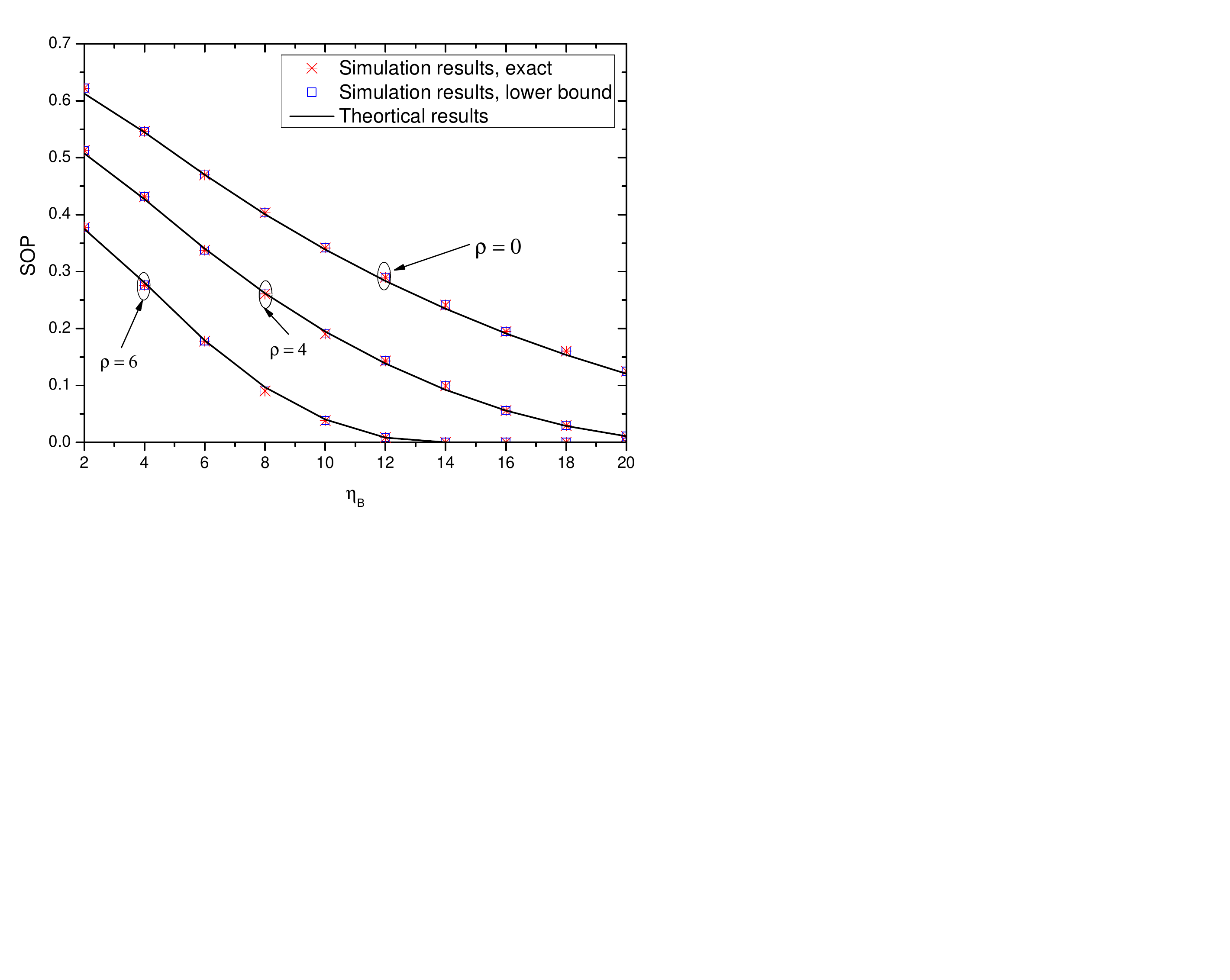}
\caption{SOP versus ${\eta _{\rm{B}}}$ with different $\rho$ when $P=60$ dB, $l=4{\rm m}$, $D = 8{\rm{m}}$, ${\gamma _{{\rm{th}}}} = 3$, and ${\eta _{\rm{E}}} = 1$.}
\label{fig9}
\end{figure}

Fig. \ref{fig10} shows SOP versus ${\eta _{\rm{B}}}$ with different ${\gamma _{{\rm{th}}}}$ when $P=60$ dB, $l = 4{\rm{m}}$, $D = 8{\rm{m}}$, ${\eta _{\rm{E}}} = 1$, and $\rho  = 4{\rm{m}}$. It can be observed that the value of SOP becomes larger as the increase of ${\gamma _{{\rm{th}}}}$, which means that the system secrecy performance is degrades. The conclusion derived in Fig. \ref{fig10} verifies the correctness of \emph{Remark \ref{rem5}}.

\begin{figure}
\centering
\includegraphics[width=8cm]{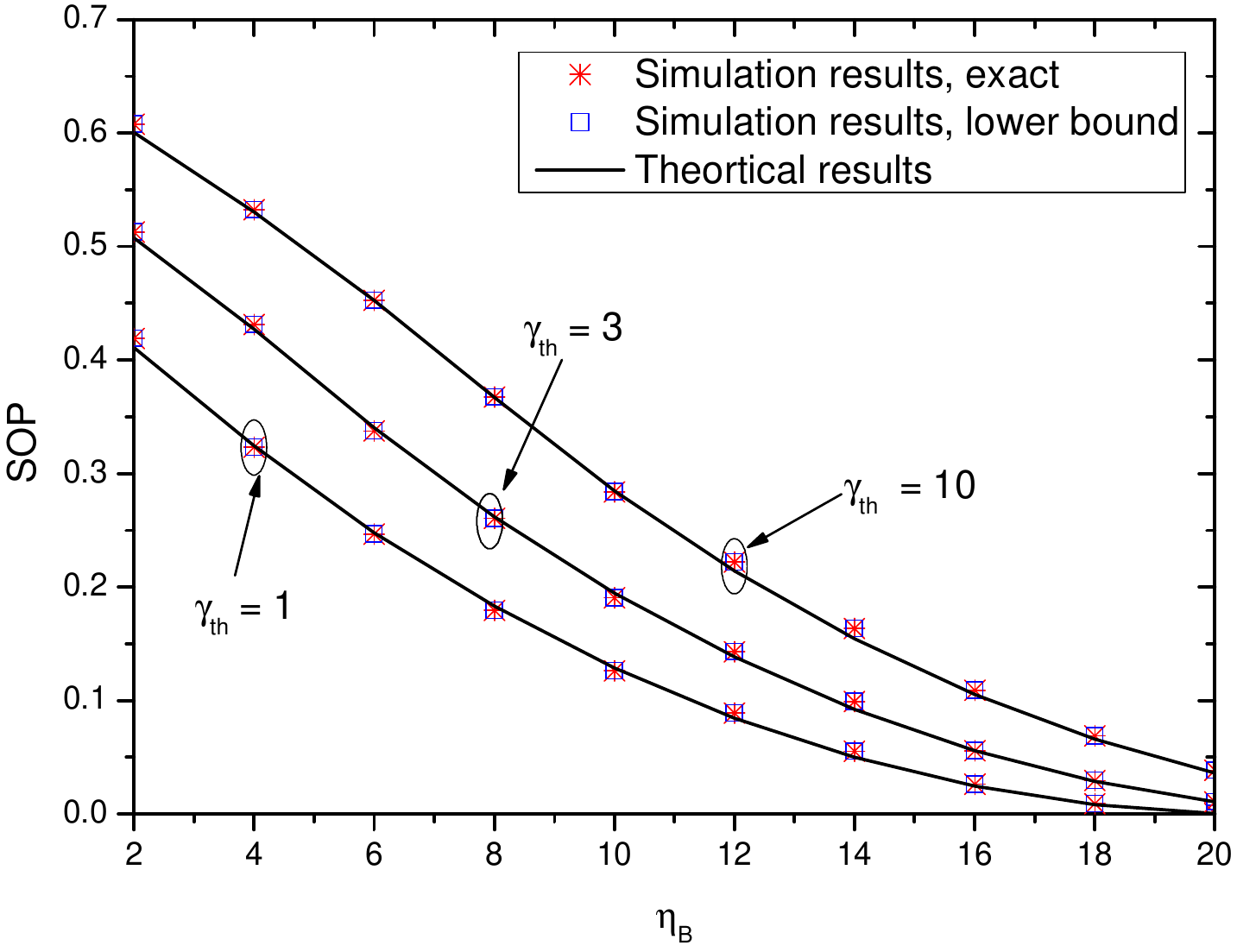}
\caption{SOP versus ${\eta _{\rm{B}}}$ with different ${\gamma _{{\rm{th}}}}$ when $P=60$ dB, $l = 4{\rm{m}}$, $D = 8{\rm{m}}$, ${\eta _{\rm{E}}} = 1$, and $\rho  = 4{\rm{m}}$.}
\label{fig10}
\end{figure}

In Fig. \ref{fig8}-Fig. \ref{fig10}, it can be observed that the theoretical values for the lower bound of SOP agree with the simulation values for the lower bound, which verifies that the theoretical analysis is correct. At the same time, the gap between the theoretical lower bound and the exact simulation result is also small, which indicates that the scaling method for deriving the lower bound of SOP is reasonable.

\section{Conclusions}
\label{section7}
This paper investigates the PLS performance for a random VLC network with Alice, Bob, and Eve. By considering the protected zone, imperfect CSIR, and the randomness of Bob and Eve, the closed-form expressions of the ASC and the lower bound of the SOP are derived, respectively. Moreover, some insights are obtained, which can be utilized for practical system design. Numerical results show that the gaps between theoretical results and simulation results are small, which indicates that the derived theoretical expressions can be used to evaluate PLS for VLC. Moreover, it is shown that the nominal optical intensity, the dimming target, the protected zone and the imperfect CSI have strong impacts on the secrecy performance. The derived theoretical expressions and results regarding parameters that influence the secrecy performance will enable system designers to quickly evaluate system performance without time-consuming simulations and determine the optimal available parameter choices when facing different security risks.

\numberwithin{equation}{section}
\appendices
\section{Proof of Theorem \ref{them1}}
\label{appa}
\renewcommand{\theequation}{A.\arabic{equation}}
According to Fig. \ref{fig1}, the projection of Alice on the receiver zone is the center of zone ${\cal S}$, and Bob is uniformly distributed in zone ${\cal S}$. Therefore, the cumulative distribution function (CDF) of ${r_{\rm{B}}}$ is given by
\begin{eqnarray}
{F_{{r_{\rm{B}}}}}({r_{\rm{B}}}) = \int_0^{2\pi } {\int_0^{{r_{\rm{B}}}} {\frac{1}{{\pi {D^2}}}r{\rm{d}}r{\rm{d}}\theta } }= \frac{{r_{\rm{B}}^2}}{{{D^2}}},0 \le {r_{\rm{B}}} \le D.
 \label{eq37}
\end{eqnarray}

By using (\ref{eq37}), the PDF of ${r_{\rm{B}}}$ is further written as
\begin{eqnarray}
{f_{{r_{\rm{B}}}}}({r_{\rm{B}}}) =\frac{{{\rm{d}}{F_{{r_{\rm{B}}}}}({r_{\rm{B}}})}}{{{\rm{d}}{r_{\rm{B}}}}} = \frac{{2{r_{\rm{B}}}}}{{{D^2}}},\;0 \le {r_{\rm{B}}} \le D.
 \label{eq38}
\end{eqnarray}

According to (\ref{eq3}) and (\ref{eq38}), the PDF of ${H_{\rm{B}}}$ can be derived as
\begin{eqnarray}
{f_{{H_{\rm{B}}}}}(h) \!\!\!\!\!& =&\! \!\!\!\!\frac{{\rm{d}}}{{{\rm{d}}h}}\!\!\left(\! { - \sqrt {{{\left( {\frac{{(m + 1)A{T_s}g{l^{m{\rm{ + 1}}}}}}{{2\pi }}} \right)}^{\!\!\frac{2}{{m + 3}}}}{h^{ - \frac{2}{{m + 3}}}} - {l^2}} }\! \right)\nonumber \\
 &\times& \!\!\!\!\! {f_{{r_{\rm{B}}}}}\!\!\left(\! {\sqrt {{{\left(\! {\frac{{(m \!+\! 1)A{T_s}g{l^{m{\rm{ + 1}}}}}}{{2\pi }}} \!\right)}^{\!\!\frac{2}{{m + 3}}}}{h^{ - \frac{2}{{m \!+\! 3}}}} \!-\! {l^2}} } \!\right)\!\!,
 \label{eq39}
\end{eqnarray}
Therefore, eq. (\ref{eq8}) holds.

Because Eve is uniformly distributed in zone ${\cal S}\backslash {\cal P}$, the CDF of ${r_{\rm{E}}}$ is given by
\begin{eqnarray}
{F_{{r_{\rm{E}}}}}({r_{\rm{E}}}) %= \int_0^{2\pi } {\int_\rho ^{{r_{\rm{E}}}} {\frac{1}{{\pi ({D^2} - {\rho ^2})}}r{\rm{d}}r{\rm{d}}\theta } }
 = \frac{{r_{\rm{E}}^2 - {\rho ^2}}}{{{D^2} - {\rho ^2}}},\rho  \le {r_{\rm{E}}} \le D.
\label{eq40}
\end{eqnarray}

By using (\ref{eq37}), the PDF of ${r_{\rm{E}}}$ is further written as
\begin{eqnarray}
{f_{{r_{\rm{E}}}}}({r_{\rm{E}}}) %= \frac{{{\rm{d}}{F_R}({r_{\rm{E}}})}}{{{\rm{d}}{r_{\rm{E}}}}}
= \frac{{2{r_{\rm{E}}}}}{{{D^2} - {\rho ^2}}},\;\rho  \le {r_{\rm{E}}} \le D.
 \label{eq41}
\end{eqnarray}

According to (\ref{eq3}) and (\ref{eq41}), the PDF of ${H_{\rm{E}}}$ can be derived as \cite{BIB24}
\begin{eqnarray}
{f_{{H_{\rm{E}}}}}(h)\!\!\!\!\!& =&\!\!\!\!\!\frac{{\rm{d}}}{{{\rm{d}}h}}\!\!\left(\!\! { - \sqrt {{{\left(\! {\frac{{(m \!+\! 1)A{T_s}g{l^{m{\rm{ + 1}}}}}}{{2\pi }}}\! \right)\!\!}^{\frac{2}{{m + 3}}}}{h^{ - \frac{2}{{m + 3}}}} - {l^2}} } \right)\nonumber \\
 &\times&\!\!\!\!\!\! {f_{{r_{\rm{E}}}}}\!\!\left(\!\! {\sqrt {{{\left(\! {\frac{{(m \!+\! 1)A{T_s}g{l^{m{\rm{ + 1}}}}}}{{2\pi }}} \!\right)\!\!}^{\frac{2}{{m + 3}}}}{h^{ - \frac{2}{{m + 3}}}} \!-\! {l^2}} }\! \right)\!\!,
\label{eq42}
\end{eqnarray}
Therefore, eq. (\ref{eq9}) can be easily derived.

\section{Proof of Theorem \ref{them2}}
\label{appb}
\renewcommand{\theequation}{B.\arabic{equation}}
Case 1: When $\chi  \le {{{v_{\min }}} \mathord{\left/
 {\vphantom {{{v_{\min }}} {{v_{\max }}}}} \right.
 \kern-\nulldelimiterspace} {{v_{\max }}}}$, by observing Fig. \ref{fig3}, $\overline {{C_s}} {\rm{ = 0}}$ is straightforward.

Case 2: When ${v_{\min }}{\rm{/}}{v_{\max }} < \chi  \le {v_{{\mathop{\rm mi}\nolimits} {\rm{d}}}}{\rm{/}}{v_{\max }}$, we have
\begin{eqnarray}
&&\!\!\!\!\!\!\!\!\!\!\!\!\!\!{C_1} \!\!=\!\!\! \underbrace {\int_{\frac{{{v_{\min }}}}{\chi }}^{{v_{{\rm{max}}}}}\!\!\! {\frac{1}{2}\!\!\ln \!\!\left[\! {\sigma _{\rm{E}}^2\!\!\left(\! {e{\xi ^2}\!{P^2}\!{\rm{1}}{{\rm{0}}^{\frac{{{\eta _{\rm{B}}}}}{5}}}\!{x^2} \!\!+\!\! 2\pi \sigma _{\rm{B}}^2} \!\right)} \!\!\right]\!\!{f_{{H_{\rm{B}}}}}\!(x)\!\!\!\int_{{v_{\min }}}^{\chi x} \!\!\!\!{{f_{{H_{\rm{E}}}}}\!(y){\rm{d}}y} } {\rm{d}}x}_{{I_1}}\nonumber \\
 &&\!\!\!\!\!\!\!\!\!\!\!\!\!\! -\!\!\!\underbrace {\int_{{v_{\min }}}^{\chi {v_{{\rm{max}}}}}\!\!\! {\frac{1}{2}\!\!\ln\!\! \left[ {2\pi \sigma _{\rm{B}}^2\!\!\left(\!\! {{\rm{1}}{{\rm{0}}^{\frac{{{\eta _{\rm{E}}}}}{5}}}\!{y^2}\!{\xi ^2}\!{P^2} \!\!+\!\! \sigma _{\rm{E}}^2} \!\!\right)}\!\! \right]\!\!{f_{{H_{\rm{E}}}}}\!(y)\!\!\!\!\int_{\frac{y}{\chi }}^{{v_{{\rm{max}}}}}\!\!\!\!\!\!\!\! {{f_{{H_{\rm{B}}}}}\!(x){\rm{d}}x{\rm{d}}y} } }_{{I_2}},
\label{eq43}
\end{eqnarray}
where $I_1$ can be written as
\begin{eqnarray}
{I_1} \!\!\!\!&=&\!\!\!\! \frac{{(m \!+\! 3){\Xi _{\rm{1}}}{\Xi _{\rm{2}}}}}{4}\!\!\left[\! v_{\min }^{ - \frac{2}{{m + 3}}}\lambda\!\! \left(\!\! {\frac{{m \!+\! 3}}{2},\frac{{{v_{\min }}}}{\chi },{v_{\max }},\frac{{e{\rm{1}}{{\rm{0}}^{\frac{{{\eta _{\rm{B}}}}}{5}}}}}{{2\pi \sigma _{\rm{B}}^2}}}\!\! \right) \right.\nonumber\\
\!&-&\!\!\!\!\left. {\chi ^{ - \frac{2}{{m + 3}}}}\lambda \!\!\left(\!\! {\frac{{m \!+\! 3}}{{\rm{4}}},\frac{{{v_{\min }}}}{\chi },{v_{\max }},\frac{{e{\rm{1}}{{\rm{0}}^{\frac{{{\eta _{\rm{B}}}}}{5}}}}}{{2\pi \sigma _{\rm{B}}^2}}}\!\! \right)\!\! \right]\!.
\label{eq49}
\end{eqnarray}

Similarly, $I_2$ in (\ref{eq43}) can be derived as
\begin{eqnarray}
{I_2} \!\!\!\!&=&\!\!\!\! \frac{{(m \!+\! 3){\Xi _{\rm{1}}}{\Xi _{\rm{2}}}}}{4}\!\!\left[\!\! {\chi ^{\frac{2}{{m + 3}}}}\lambda\!\! \left(\!\! {\frac{{m \!+\! 3}}{{\rm{4}}},{v_{\min }},\chi {v_{\max }},\frac{{{\rm{1}}{{\rm{0}}^{\frac{{{\eta _{\rm{E}}}}}{5}}}}}{{\sigma _{\rm{E}}^2}}}\!\! \right) \right.\nonumber\\
\!&-&\!\!\!\! \left.v_{{\rm{max}}}^{ - \frac{2}{{m + 3}}}\lambda\!\! \left(\!\! {\frac{{m \!+\! 3}}{{\rm{2}}},{v_{\min }},\chi {v_{\max }},\frac{{{\rm{1}}{{\rm{0}}^{\frac{{{\eta _{\rm{E}}}}}{5}}}}}{{\sigma _{\rm{E}}^2}}}\!\! \right)\!\! \right]\!.
 \label{eq55}
\end{eqnarray}
By submitting (\ref{eq49}) and (\ref{eq55}) into (\ref{eq43}), $C_1$ can be derived.

Case 3: when ${v_{{\mathop{\rm mid}\nolimits} }}{\rm{/}}{v_{\max }} < \chi  \le 1$, we can get (\ref{eq56}) as shown at the top of the next page,
\begin{table*}\normalsize
\begin{eqnarray}
{C_{\rm{2}}} \!\!\!\!\!\!&=&\!\!\!\!\!\!\! \frac{{{\Xi _{\rm{1}}}\!{\Xi _{\rm{2}}}}}{2}\!\!\!\left\{\!\! {\underbrace {\int_{\frac{{{v_{{\rm{min}}}}}}{\chi }}^{\frac{{{v_{{\rm{mid}}}}}}{\chi }}\!\!\! {\ln\!\! \left[\! {\sigma _{\rm{E}}^2\!\!\left(\! {e{\xi ^2}{P^2}{\rm{1}}{{\rm{0}}^{\frac{{{\eta _{\rm{B}}}}}{5}}}{x^2} \!+\! 2\pi \sigma _{\rm{B}}^2} \!\!\right)}\!\! \right]\!\!{x^{ - \frac{m\!+\!5}{{m \!+\! 3}} }}\!\!\!\int_{{v_{\min }}}^{\chi x}\!\!\! {{y^{ - \frac{m\!+\!5}{{m \!+\! 3}} }}{\rm{d}}y} } {\rm{d}}x}_{{D_1}}} \right. \!+\! \underbrace {\int_{\frac{{{v_{{\rm{mid}}}}}}{\chi }}^{{v_{{\rm{max}}}}}\!\!\! {\ln\!\! \left[\! {\sigma _{\rm{E}}^2\!\!\left(\!\! {e{\xi ^2}{P^2}{\rm{1}}{{\rm{0}}^{\frac{{{\eta _{\rm{B}}}}}{5}}}{x^2} \!+\! 2\pi \sigma _{\rm{B}}^2} \!\!\right)}\!\! \right]\!\!{x^{ - \frac{m\!+\!5}{{m \!+\! 3}} }}\!\!\!\int_{{v_{\min }}}^{{v_{{\rm{mid}}}}}\!\!\!\!\! {{y^{ - \frac{m\!+\!5}{{m \!+\! 3}} }}{\rm{d}}y} } {\rm{d}}x}_{{D_2}}\nonumber \\
&-&\!\!\!\!\left. {  \underbrace {\int_{{v_{\min }}}^{{v_{{\rm{mid}}}}} {\ln \left[ {2\pi \sigma _{\rm{B}}^2\left( {{\rm{1}}{{\rm{0}}^{\frac{{{\eta _{\rm{E}}}}}{5}}}{y^2}{\xi ^2}{P^2} + \sigma _{\rm{E}}^2} \right)} \right]{y^{ - \frac{m\!+\!5}{{m + 3}}}}\int_{\frac{y}{\chi }}^{{v_{{\rm{max}}}}} {{x^{ - \frac{m\!+\!5}{{m + 3}} }}{\rm{d}}x{\rm{d}}y} } }_{{D_3}}} \right\},
\label{eq56}
\end{eqnarray}
\hrulefill
\end{table*}
where $D_1$ can be written as
\begin{eqnarray}
{D_1}\!\!&=&\!\! \frac{{m \!+\! 3}}{2}\!\!\left[\! v_{\min }^{ - \frac{2}{{m + 3}}}\lambda\!\! \left(\!\! {\frac{{m \!+\! 3}}{2},\frac{{{v_{\min }}}}{\chi },\frac{{{v_{{\rm{mid}}}}}}{\chi },\frac{{e{\rm{1}}{{\rm{0}}^{\frac{{{\eta _{\rm{B}}}}}{5}}}}}{{2\pi \sigma _{\rm{B}}^2}}}\!\! \right) \right. \nonumber\\
\!&-&\! \left.{\chi ^{ - \frac{2}{{m + 3}}}}\lambda\!\! \left(\!\! {\frac{{m \!+\! 3}}{4},\frac{{{v_{\min }}}}{\chi },\frac{{{v_{{\rm{mid}}}}}}{\chi },\frac{{e{\rm{1}}{{\rm{0}}^{\frac{{{\eta _{\rm{B}}}}}{5}}}}}{{2\pi \sigma _{\rm{B}}^2}}} \!\!\right)\!\! \right]\!.
\label{eq58}
\end{eqnarray}

For $D_2$ in (\ref{eq56}), we have
\begin{eqnarray}
{D_2} \!\!=\!\!\frac{{m \!+\! 3}}{2}\!\!\left(\!\! {v_{\min }^{ - \frac{2}{{m \!+\! 3}}} \!\!-\!\! v_{{\mathop{\rm mi}\nolimits} {\rm{d}}}^{ - \frac{2}{{m \!+\! 3}}}} \!\!\right)\!\lambda\!\! \left(\!\! {\frac{{m \!+\! 3}}{2},\frac{{{v_{{\rm{mid}}}}}}{\chi },{v_{\max }},\frac{{e{\rm{1}}{{\rm{0}}^{\frac{{{\eta _{\rm{B}}}}}{5}}}}}{{2\pi \sigma _{\rm{B}}^2}}} \!\!\right)\!\!.
\label{eq59}
\end{eqnarray}

For $D_3$ in (\ref{eq56}), we have
\begin{eqnarray}
{D_3} \!\!&=&\!\! \frac{{m \!+\! 3}}{2}\!\!\left[\! {\chi ^{\frac{2}{{m + 3}}}}\lambda\!\! \left(\!\! {\frac{{m \!+\! 3}}{4},{v_{\min }},{v_{{\rm{mid}}}},\frac{{{\rm{1}}{{\rm{0}}^{\frac{{{\eta _{\rm{E}}}}}{5}}}}}{{\sigma _{\rm{E}}^2}}}\!\! \right)\right. \nonumber\\
\!&-&\! \left.v_{{\rm{max}}}^{ - \frac{2}{{m + 3}}}\lambda\!\! \left(\!\! {\frac{{m \!+\! 3}}{2},{v_{\min }},{v_{{\rm{mid}}}},\frac{{{\rm{1}}{{\rm{0}}^{\frac{{{\eta _{\rm{E}}}}}{5}}}}}{{\sigma _{\rm{E}}^2}}}\!\! \right)\!\! \right]\!.
\label{eq61}
\end{eqnarray}
By submitting (\ref{eq58}), (\ref{eq59}) and (\ref{eq61}) into (\ref{eq56}), $C_2$ can be derived.

Case 4: When $1 < \chi  \le {v_{{\mathop{\rm mid}\nolimits} }}{\rm{/}}{v_{\min }}$, we have
\begin{eqnarray}
{C_3} &=& \underbrace {\int_{{v_{\min }}}^{\frac{{{v_{{\mathop{\rm mid}\nolimits} }}}}{\chi }} {{f_{{H_{\rm{B}}}}}(x)\int_{{v_{\min }}}^{\chi x} {C_{s,2}^{{\rm{Low}}}{f_{{H_{\rm{E}}}}}(y){\rm{d}}y} } {\rm{d}}x}_{{A_1}} \nonumber\\
&+& \underbrace {\int_{\frac{{{v_{{\mathop{\rm m}\nolimits} {\rm{id}}}}}}{\chi }}^{{v_{{\rm{max}}}}} {{f_{{H_{\rm{B}}}}}(x)\int_{{v_{\min }}}^{{v_{{\rm{mid}}}}} {C_{s,2}^{{\rm{Low}}}{f_{{H_{\rm{E}}}}}(y){\rm{d}}y} } {\rm{d}}x}_{{A_2}},
\label{eq62}
\end{eqnarray}
where $A_1$ and $A_2$ can be written as
\begin{eqnarray}
&&\!\!\!\!\!\!\!\!\!\!\!\!\!\!{A_1} \!\!=\!\! \frac{{(m \!\!+\!\! 3){\Xi _{\rm{1}}}{\Xi _{\rm{2}}}}}{4}\!\!\left\{\!\! v_{\min }^{ - \frac{2}{{m \!+\! 3}}}\lambda\!\! \left(\!\! {\frac{{m \!+\! 3}}{2},\!{v_{\min }},\!\frac{{{v_{{\mathop{\rm mid}\nolimits} }}}}{\chi },\!\frac{{e{\rm{1}}{{\rm{0}}^{\frac{{{\eta _{\rm{B}}}}}{5}}}}}{{2\pi \sigma _{\rm{B}}^2}}} \!\!\right) \!-\! {\chi ^{ - \frac{2}{{m \!+\! 3}}}} \right.\nonumber\\
&&\!\!\!\!\!\!\!\!\!\!\!\!\!\! \times \lambda\!\! \left(\!\! {\frac{{m \!+\! 3}}{4},{v_{\min }},\frac{{{v_{{\mathop{\rm mid}\nolimits} }}}}{\chi },\frac{{e{\rm{1}}{{\rm{0}}^{\frac{{{\eta _{\rm{B}}}}}{5}}}}}{{2\pi \sigma _{\rm{B}}^2}}}\!\! \right)-\left[\! {v_{\min }^{ - \frac{2}{{m + 3}}} \!-\! {{\left(\!\! {\frac{{{v_{{\rm{mid}}}}}}{\chi }} \!\!\right)\!}^{ - \frac{2}{{m + 3}}}}}\! \right] \nonumber \\
 &&\!\!\!\!\!\!\!\!\!\!\!\!\!\!\times \!\lambda \!\!\left(\!\! {\frac{{m \!\!+\!\! 3}}{2},\!{v_{\min }},\!\chi {v_{\min }},\!\frac{{{\rm{1}}{{\rm{0}\!}^{\frac{{{\eta _{\rm{E}}}}}{5}}}}}{{\sigma _{\rm{E}}^2}}}\!\! \right) \!\!-\! {\chi ^{\frac{2}{{m \!+\! 3}}}}\!\lambda\!\! \left(\!\! {\frac{{m \!\!+\!\! 3}}{4},\!\chi {v_{\min }},\!{v_{{\rm{mid}}}},\!\frac{{{\rm{1}}{{\rm{0}\!}^{\frac{{{\eta _{\rm{E}}}}}{5}}}}}{{\sigma _{\rm{E}}^2}}}\!\! \right)\nonumber\\
 &&\!\!\!\!\!\!\!\!\!\!\!\!\!\!+\left. {{{\left( {\frac{{{v_{{\rm{mid}}}}}}{\chi }} \right)}^{ - \frac{2}{{m \!+\! 3}}}}\lambda \left(\!\! {\frac{{m \!+\! 3}}{2},\chi {v_{\min }},{v_{{\rm{mid}}}},\frac{{{\rm{1}}{{\rm{0}}^{\frac{{{\eta _{\rm{E}}}}}{5}}}}}{{\sigma _{\rm{E}}^2}}} \!\!\right)}\!\! \right\}.
\label{eq67}
\end{eqnarray}
and
\begin{eqnarray}
&&\!\!\!\!\!\!\!\!\!\!\!\!{A_2} \!\!=\!\! \frac{{(m \!\!+\!\! 3){\Xi _{\rm{1}}}\!{\Xi _{\rm{2}}}}}{{\rm{4}}}\!\!\!\left\{\!\!\! {\left(\!\! {v_{\min }^{ - \frac{2}{{m \!+\! 3}}} \!-\! v_{{\rm{mid}}}^{ - \frac{2}{{m \!+\! 3}}}}\!\! \right)\!\lambda\!\! \left(\!\! {\frac{{m \!\!+\!\! 3}}{2},\!\frac{{{v_{{\rm{mid}}}}}}{\chi },\!{v_{\max }},\!\frac{{e\!{\rm{1}}{{\rm{0}}^{\frac{{{\eta _{\rm{B}}}}}{5}}}}}{{2\pi \sigma _{\rm{B}}^2}}}\!\! \right)} \right.\nonumber \\
&&\!\!\!\!\!\!\!\!\!\!\!-\!\left.{  \left[\!\! {{{\left(\!\! {\frac{{{v_{{\rm{mid}}}}}}{\chi }}\!\! \right)}^{ - \frac{2}{{m \!+\! 3}}}} \!\!-\! v_{\max }^{ - \frac{2}{{m \!+\! 3}}}} \!\!\right]\!\lambda\!\! \left(\!\! {\frac{{m \!+\! 3}}{2},{v_{\min }},{v_{{\rm{mid}}}},\frac{{{\rm{1}}{{\rm{0}}^{\frac{{{\eta _{\rm{E}}}}}{5}}}}}{{\sigma _{\rm{E}}^2}}} \!\!\right)}\!\! \right\}.
\label{eq71}
\end{eqnarray}

Submitting (\ref{eq67}) and (\ref{eq71}) into (\ref{eq62}), $C_3$ can be finally derived.

Case 5: When $\chi  > {v_{{\mathop{\rm mid}\nolimits} }}{\rm{/}}{v_{\min }}$, we have
\begin{equation}
{C_4}{\rm{ = }}{\left. {{A_2}} \right|_{\chi  = \frac{{{v_{{\mathop{\rm mid}\nolimits} }}}}{{{v_{\min }}}}}}
\label{eq72}
\end{equation}
After some simplifications, $C_4$ can be finally derived.

\section{Proof of Theorem \ref{them3}}
\label{appc}
\renewcommand{\theequation}{C.\arabic{equation}}
Case 1: When ${\gamma _{{\rm{th}}}} \le {\chi ^2}\frac{{v_{\min }^2}}{{v_{{\rm{mid}}}^2}}$, we have $P_{{\rm{SOP}}}^{\rm{L}} = 0$.

Case 2: When ${\chi ^2}\frac{{v_{\min }^2}}{{v_{{\rm{mid}}}^2}} \le {\gamma _{{\rm{th}}}} \le {\chi ^2}$, we have
\begin{eqnarray}
P_{{\rm{SOP}}}^{\rm{L}} %&=& \frac{{{\Xi _1}{\Xi _{\rm{2}}}}}{4}{\left( {\frac{{2\pi \sigma _{\rm{B}}^2\sigma _{\rm{E}}^2}}{{{\rm{1}}{{\rm{0}}^{\frac{{{\eta _{\rm{E}}} + {\eta _{\rm{B}}}}}{5}}}e{\xi ^4}{P^4}}}} \right)^{ - \frac{1}{{m + 3}}}}\int_{\frac{{{J_{{\rm{B,min}}}}}}{{{\gamma _{{\rm{th}}}}}}}^{{J_{{\rm{E,max}}}}} {{z^{ - \frac{1}{{m + 3}} - 1}}\int_{{J_{{\rm{B,min}}}}}^{{\gamma _{{\rm{th}}}}z} {{y^{ - \frac{1}{{m + 3}} - 1}}{\rm{d}}y{\rm{d}}z} } \nonumber \\
\!\!\!\!\! &=&\!\!\!\!\! \frac{{{\Xi _1}{\Xi _{\rm{2}}}}}{4}{\left( {\frac{{2\pi \sigma _{\rm{B}}^2\sigma _{\rm{E}}^2}}{{{\rm{1}}{{\rm{0}}^{\frac{{{\eta _{\rm{E}}} + {\eta _{\rm{B}}}}}{5}}}e{\xi ^4}{P^4}}}} \right)^{ - \frac{1}{{m + 3}}}}{(m + 3)^2}\nonumber\\
&\times&\!\!\!\!\!J_{{\rm{B,min}}}^{ - \frac{1}{{m + 3}}}\left[ {{{\left( {\frac{{{J_{{\rm{B,min}}}}}}{{{\gamma _{{\rm{th}}}}}}} \right)}^{ - \frac{1}{{m + 3}}}} - J_{{\rm{E,max}}}^{ - \frac{1}{{m + 3}}}} \right]\nonumber \\
 &-& \frac{{{\Xi _1}{\Xi _{\rm{2}}}}}{8}{\left( {\frac{{2\pi \sigma _{\rm{B}}^2\sigma _{\rm{E}}^2}}{{{\rm{1}}{{\rm{0}}^{\frac{{{\eta _{\rm{E}}} + {\eta _{\rm{B}}}}}{5}}}e{\xi ^4}{P^4}}}} \right)^{ - \frac{1}{{m + 3}}}}{(m + 3)^2}\nonumber\\
 &\times&\!\!\!\!\!\gamma _{{\rm{th}}}^{ - \frac{1}{{m + 3}}}\left[ {{{\left( {\frac{{{J_{{\rm{B,min}}}}}}{{{\gamma _{{\rm{th}}}}}}} \right)}^{ - \frac{2}{{m + 3}}}} - J_{{\rm{E,max}}}^{ - \frac{2}{{m + 3}}}} \right].
 \label{eq73}
\end{eqnarray}

Case 3: When ${\chi ^2} \le {\gamma _{{\rm{th}}}} \le {\chi ^2}\frac{{v_{\max }^2}}{{v_{{\rm{mid}}}^2}}$, we have
\begin{eqnarray}
P_{{\rm{SOP}}}^{\rm{L}} %&=& \frac{{{\Xi _1}{\Xi _{\rm{2}}}}}{4}{\left( {\frac{{2\pi \sigma _{\rm{B}}^2\sigma _{\rm{E}}^2}}{{{\rm{1}}{{\rm{0}}^{\frac{{{\eta _{\rm{E}}} + {\eta _{\rm{B}}}}}{5}}}e{\xi ^4}{P^4}}}} \right)^{ - \frac{1}{{m + 3}}}}\int_{{J_{{\rm{E,min}}}}}^{{J_{{\rm{E,max}}}}} {{z^{ - \frac{1}{{m + 3}} - 1}}\int_{{J_{{\rm{B,min}}}}}^{{\gamma _{{\rm{th}}}}z} {{y^{ - \frac{1}{{m + 3}} - 1}}{\rm{d}}y{\rm{d}}z} } \nonumber \\
\!\!\!\!\! &=&\!\!\!\!\! \frac{{{\Xi _1}{\Xi _{\rm{2}}}}}{4}{\left( {\frac{{2\pi \sigma _{\rm{B}}^2\sigma _{\rm{E}}^2}}{{{\rm{1}}{{\rm{0}}^{\frac{{{\eta _{\rm{E}}} + {\eta _{\rm{B}}}}}{5}}}e{\xi ^4}{P^4}}}} \right)^{ - \frac{1}{{m + 3}}}}{(m + 3)^2}\nonumber\\
&\times&\!\!\!\!\!J_{{\rm{B,min}}}^{ - \frac{1}{{m + 3}}}\left[ {J_{{\rm{E,min}}}^{ - \frac{1}{{m + 3}}} - J_{{\rm{E,max}}}^{ - \frac{1}{{m + 3}}}} \right]\nonumber \\
 &-& \!\!\!\!\!\frac{{{\Xi _1}{\Xi _{\rm{2}}}}}{8}{\left( {\frac{{2\pi \sigma _{\rm{B}}^2\sigma _{\rm{E}}^2}}{{{\rm{1}}{{\rm{0}}^{\frac{{{\eta _{\rm{E}}} + {\eta _{\rm{B}}}}}{5}}}e{\xi ^4}{P^4}}}} \right)^{ - \frac{1}{{m + 3}}}}{(m + 3)^2}\nonumber\\
 &\times&\gamma _{{\rm{th}}}^{ - \frac{1}{{m + 3}}}\left[ {J_{{\rm{E,min}}}^{ - \frac{2}{{m + 3}}} - J_{{\rm{E,max}}}^{ - \frac{2}{{m + 3}}}} \right].
\label{eq74}
\end{eqnarray}

Case 4: When ${\chi ^2}\frac{{v_{\max }^2}}{{v_{{\rm{mid}}}^2}} \le {\gamma _{{\rm{th}}}} \le {\chi ^2}\frac{{v_{\max }^2}}{{v_{\min }^2}}$, we can get
\begin{eqnarray}
P_{{\rm{SOP}}}^{\rm{L}} \!\!\!\!\!%&=&\!\!\!\! \frac{{{\Xi _1}{\Xi _{\rm{2}}}}}{4}{\left( {\frac{{2\pi \sigma _{\rm{B}}^2\sigma _{\rm{E}}^2}}{{{\rm{1}}{{\rm{0}}^{\frac{{{\eta _{\rm{E}}} + {\eta _{\rm{B}}}}}{5}}}e{\xi ^4}{P^4}}}} \right)^{ - \frac{1}{{m + 3}}}}\int_{{J_{{\rm{E,min}}}}}^{\frac{{{J_{{\rm{B,max}}}}}}{{{\gamma _{{\rm{th}}}}}}} {{z^{ - \frac{1}{{m + 3}} - 1}}\int_{{J_{{\rm{B,min}}}}}^{{\gamma _{{\rm{th}}}}z} {{y^{ - \frac{1}{{m + 3}} - 1}}{\rm{d}}y{\rm{d}}z} } \nonumber \\
% &+&\!\!\!\! \frac{{{\Xi _1}{\Xi _{\rm{2}}}}}{4}{\left( {\frac{{2\pi \sigma _{\rm{B}}^2\sigma _{\rm{E}}^2}}{{{\rm{1}}{{\rm{0}}^{\frac{{{\eta _{\rm{E}}} + {\eta _{\rm{B}}}}}{5}}}e{\xi ^4}{P^4}}}} \right)^{ - \frac{1}{{m + 3}}}}\int_{\frac{{{J_{{\rm{B,max}}}}}}{{{\gamma _{{\rm{th}}}}}}}^{{J_{{\rm{E,max}}}}} {{z^{ - \frac{1}{{m + 3}} - 1}}\int_{{J_{{\rm{B,min}}}}}^{{J_{{\rm{B,max}}}}} {{y^{ - \frac{1}{{m + 3}} - 1}}{\rm{d}}y{\rm{d}}z} } \nonumber \\
 &=&\!\!\!\!\! \frac{{{\Xi _1}{\Xi _{\rm{2}}}}}{4}{\left( {\frac{{2\pi \sigma _{\rm{B}}^2\sigma _{\rm{E}}^2}}{{{\rm{1}}{{\rm{0}}^{\frac{{{\eta _{\rm{E}}} + {\eta _{\rm{B}}}}}{5}}}e{\xi ^4}{P^4}}}} \right)^{ - \frac{1}{{m + 3}}}}{(m + 3)^2}\nonumber\\
 &\times& \!\!\!\!\! J_{{\rm{B,min}}}^{ - \frac{1}{{m + 3}}}\left[ {J_{{\rm{E,min}}}^{ - \frac{1}{{m + 3}}} - {{\left( {\frac{{{J_{{\rm{B,max}}}}}}{{{\gamma _{{\rm{th}}}}}}} \right)}^{ - \frac{1}{{m + 3}}}}} \right]\nonumber \\
 &-&\!\!\!\! \frac{{{\Xi _1}{\Xi _{\rm{2}}}}}{8}{\left( {\frac{{2\pi \sigma _{\rm{B}}^2\sigma _{\rm{E}}^2}}{{{\rm{1}}{{\rm{0}}^{\frac{{{\eta _{\rm{E}}} + {\eta _{\rm{B}}}}}{5}}}e{\xi ^4}{P^4}}}} \right)^{ - \frac{1}{{m + 3}}}}{(m + 3)^2}\nonumber\\
 &\times& \!\!\!\!\!\gamma _{{\rm{th}}}^{ - \frac{1}{{m + 3}}}\left[ {J_{{\rm{E,min}}}^{ - \frac{2}{{m + 3}}} - {{\left( {\frac{{{J_{{\rm{B,max}}}}}}{{{\gamma _{{\rm{th}}}}}}} \right)}^{ - \frac{2}{{m + 3}}}}} \right]\nonumber \\
 &+&\!\!\!\! \frac{{{\Xi _1}{\Xi _{\rm{2}}}}}{4}{\left(\! {\frac{{2\pi \sigma _{\rm{B}}^2\sigma _{\rm{E}}^2}}{{{\rm{1}}{{\rm{0}}^{\frac{{{\eta _{\rm{E}}} + {\eta _{\rm{B}}}}}{5}}}e{\xi ^4}{P^4}}}} \!\right)^{\! - \frac{1}{{m + 3}}}}{(m \!+\! 3)^2}\!\!\nonumber\\
 &\times&\!\!\!\!\!\!\left[\! {J_{{\rm{B,min}}}^{ - \frac{1}{{m + 3}}} \!-\! J_{{\rm{B,max}}}^{ - \frac{1}{{m + 3}}}} \!\right]\!\!\left[\!\! {{{\left(\! {\frac{{{J_{{\rm{B,max}}}}}}{{{\gamma _{{\rm{th}}}}}}} \!\!\right)}^{\!\!\! - \frac{1}{{m \!+\! 3}}}} \!-\! J_{{\rm{E,max}}}^{ - \frac{1}{{m + 3}}}}\!\! \right]\!\!.
\label{eq75}
\end{eqnarray}

Case 5: When ${\gamma _{{\rm{th}}}} \ge {\chi ^2}\frac{{v_{\max }^2}}{{v_{\min }^2}}$, we have $P_{{\rm{SOP}}}^{\rm{L}} = 1$.

% that's all folks
\end{document}